 \documentclass[3p,times]{elsarticle}
\usepackage{amssymb}
\usepackage{graphicx}% Include figure files
\usepackage{dcolumn}% Align table columns on decimal point
\usepackage{bm}% bold math
\usepackage{mathptmx}
\usepackage{enumerate}
\usepackage{array,color}
\usepackage{amsmath}
\usepackage{hyperref}
\usepackage{amsthm}
\newcommand{\ii}{\mathrm{i}}
\newcommand{\ee}{\mathrm{e}}

\newcommand{\diag}{\mathrm{diag}}

\newcommand{\cO}{\mathcal{O}}

\newtheorem{theorem}{Theorem}
\newtheorem{lemma}{Lemma}
\newtheorem{corol}{Corollary}
\newtheorem{prop}{Proposition}
\newtheorem{propo}{Proposition}[section]
\newtheorem{definition}{Definition}[section]

\journal{ }

\begin{document}

\begin{frontmatter}



\title{Coexistence of two distinct rogue wave patterns in the coupled nonlinear Schr\"odinger equation}

\author[aff1]{Zixuan Deng}
%\ead{zxdengmath@163.com}

\author[aff1]{Huian Lin$^{*}$}
\ead{linhamath@163.com}

\author[aff1]{Liming Ling}
%\ead{linglm@scut.edu.cn}

%\affiliation[aff1]{organization={School of Mathematics, South China University of Technology},city={Guangzhou},postcode={510641}, country={China}}
%\address[aff1]{Department of Mathematics and Physics, North China Electric Power University, Baoding, 071003, China}	

\address[aff1]{School of Mathematics, South China University of Technology, Guangzhou, 510641, China}	
%\author{}

%\affiliation{organization={},%Department and Organization
	%           addressline={}, 

\begin{abstract}
	This paper investigates the asymptotic behavior of high-order vector rogue wave (RW) solutions of the coupled nonlinear Schr\"odinger (CNLS) equation in the presence of multiple large internal parameters. We report several new high-order RW patterns in the CNLS system, including double-sector, double-heart, and mixed sector-heart configurations. The main novelty is that each RW pattern contains two distinct regions in which two different fundamental first-order RWs coexist simultaneously, potentially appearing as bright (eye-shaped) versus four-petaled or dark (anti-eye-shaped) forms. These two regions are respectively associated with the simple root structures of two different Adler--Moser polynomials: each region consists of well-separated first-order RWs in one-to-one correspondence with the simple roots of the associated polynomial. In addition, by tuning certain free parameters, the two regions of the RW pattern can be shifted to arbitrary locations in the $ (x,t) $-plane. This flexibility, together with the rich simple-root structures of Adler--Moser polynomials, enables the systematic generation of a much broader family of structured RW patterns in the CNLS equation.
	
\end{abstract}




\end{frontmatter}



\section{Introduction}\label{Sec-intro}

Rogue waves (RWs), also known as freak waves, extreme waves, or monster waves, are a class of nonlinear wave phenomena in nature characterized by extreme localization, short-lived occurrences, and high amplitudes \cite{GuoB2017book,Akhmediev2009}. This concept was first introduced by Draper in 1965 \cite{Draper1965}. Over subsequent decades, the suddenness and destructive impact of RWs have been repeatedly implicated in maritime disaster reports, posing serious threats to vessel navigation, offshore oil and gas facilities and personnel, and even coastal populations \cite{Kjeldsen1984,Kharif2003,Rosenthal2005}. Due to early observational conditions that made it difficult to obtain reproducible and verifiable scientific evidence, RWs were once regarded as rare and inexplicable accidents. In 1995, however, a RW striking an offshore oil platform in the Norwegian North Sea was recorded in full for the first time. As it occurred on New Year’s Day, it became known as the ``New Year Wave" \cite{Walker2004}. 

Over the past two decades, RW phenomena have been identified far beyond oceanography, appearing across many areas of the natural sciences. For instance, in $2007$, Solli et al. reported the first experimental observation of optical RW with abnormally large amplitudes in nonlinear optics\cite{Solli2007}. In $2010$, Kibler et al. observed the Peregrine soliton (the first-order RW) in nonlinear fiber experiments and performed corresponding numerical simulations \cite{Kibler2010}. In the same year, Efimov et al. observed RW events in superfluid helium, noting that they arise shortly after the driving is switched on, during the nonequilibrium stage before a steady state is reached \cite{2010Efimov}. Stenflo and Marklund also investigated the possibility of RWs in the atmosphere \cite{Stenflo2010}. Yan analytically derives explicit ``financial RW" solutions within the Ivancevic nonlinear option-pricing model and argues that these solutions can serve as stylized descriptions of abrupt, extreme events and their possible mechanisms in financial markets \cite{Yan2010f}. In 2011, Chabchoub et al. observed the Peregrine soliton in a wave-tank experiment \cite{Chabchoub2011}, and in 2012, further obtained higher-order RWs in a water-tank setting consistent with theoretical predictions \cite{Chabchoub2012}. These advances have substantially intensified academic interest in RWs. 

A variety of physical mechanisms have been proposed for RW generation, including  modulational instability \cite{shats2010, Pelinovsky2016}, three-dimensional directional wave focusing \cite{fochesato2007}, and so on. From a mathematical perspective, RW can be modeled as special solutions of integrable nonlinear partial differential equations, often referred to as RW solutions. Representative examples include scalar integrable nonlinear equations such as the nonlinear Schr\"odinger (NLS) equation, the derivative NLS (DNLS) equation, and the Ablowitz-Ladik equation, as well as multi-component integrable systems such as the Manakov system (i.e., coupled NLS equation, CNLS equation), the coupled Fokas-Lenells equation, and multi-component NLS equations \cite{ankiewicz2010d,ling2014high,GuoB2017book,YangB2024book, suh2023,linh2024b}.

In 2013, Kedziora et al. discovered that higher-order RW solutions of the NLS equation exhibit diverse patterns composed of multiple fundamental first-order and lower-order RWs, thereby unveiling underlying symmetry principles and extrapolating the general shape for RW solutions beyond order $6$ \cite{Kedziora2013}. Since then, RW patterns have attracted sustained attention and have become a central topic in RW theory, as understanding these patterns enables the prediction of future RW structures from initial waveforms. Notably, in 2021, Zhang et al. discovered that the structures of the RW patterns in integrable systems are governed by certain special polynomials, which they called ``governing polynomials" \cite{ZhangGQ2021}. Meanwhile, Yang et al. discovered an intrinsic connection between RW patterns of the NLS equation (with a large internal parameter) and a special class of governing polynomial, namely, the Yablonskii--Vorob'ev polynomial hierarchy. Furthermore, in 2022, Dong et al. investigated the lump patterns in the KPI equation including those in the large-time regime and at $t=0$, by using the KP constraint theory \cite{DongJ2022}. They established a relation between these lump patterns and the RW patterns of the AKNS system, as well as the Yablonskii--Vorob'ev polynomial hierarchy. With further developments, researchers discovered higher-order RW patterns (containing a single large internal parameter) associated with the root structures of the Okamoto polynomial hierarchy and the generalized Wronskian--Hermite polynomial hierarchy in multi-component integrable systems, with no more than four components \cite{YangB2024book, zhangg2022,suh2023}. Building on these works, in 2024, we reported RW patterns related to generalized Wronskian--Hermite polynomials in the multi-component DNLS equation, where the number of components is an arbitrary positive integer and the high-order RW solutions contain a single large internal parameter \cite{linh2024a}. 

In addition, for high-order RW solutions in integrable systems with multiple large internal parameters, both Yang and our group independently investigated RW patterns of the scalar NLS equation associated with the Adler--Moser polynomials \cite{yang2024,YangB2025,linh2024b}. Our group also reported RW patterns of multi-component NLS equation associated with the generalized mixed Adler--Moser polynomials \cite{lin2025vector}. However, in that study of Ref. \cite{lin2025vector}, we considered only the special case where the associated characteristic polynomial of $ (n+1) $-degree possesses an $ (n+1) $-multiple root. For a polynomial of $ (n+1) $-degree having an $ (n+1) $-multiple root is an exceptionally restrictive and highly non-generic constraint. It is therefore natural to investigate more general scenarios for the characteristic polynomial, such as the presence of double roots. 

On the other hand, the results in Ref.~\cite{ZhaoLing2016} indicate that three distinct types of fundamental first-order RWs may arise in high-order RW solutions of multi-component integrable systems, namely bright (eye-shaped) RWs, dark (anti-eye-shaped) RWs, and four-petaled RWs. However, in previous studies on RW patterns of multi-component integrable systems, it has been observed that the RW pattern in each component of a high-order vector RW solution typically consists of only one of these three types of first-order RWs. This naturally leads to the following question: can the RW pattern in a single component of a high-order vector RW solution in a multi-component integrable system simultaneously involve first-order RWs with different structures?

Therefore, in this paper, we consider the following CNLS equation: 
\begin{equation}\label{cnlse}
	\begin{aligned}
		&\ii \mathbf{q}_{t} +\frac{1}{2}\mathbf{q}_{xx}+ \mathbf{q}\mathbf{q}^{\dagger}\mathbf{q} ={0},
	\end{aligned}
\end{equation}
where $ \mathbf{q}=\left[q_{1}(x,t), q_{2}(x,t) \right]^{T} $ is a two-dimensional complex vector field, and the superscripts $ \dagger $ and $ T $ denote the conjugate transpose and transpose of the matrix, respectively.
Then, we investigate the patterns of the high-order RW solutions with multiple large internal parameters when the corresponding characteristic equation 
\begin{equation}\label{chareq1}
	\left( \chi-2{\lambda}-\sum_{i=1}^{2}\frac{a_{i}^{2}}{\chi+b_{i}}\right) \prod_{j=1}^{2}(\chi+b_{j})=0, \quad \lambda\in\mathbb{C},\quad a_{i}\in \mathbb{R}\setminus\{0\}, \quad b_{i}\in \mathbb{R}, \quad i=1,2,
\end{equation}
admits double roots. Since such double roots are determined by the algebraic equation
\begin{equation}\label{chareq2}
	 1+\sum_{i=1}^{2}\frac{a_{i}^{2}}{(\chi+b_{i})^{2}}=0,
\end{equation}
with the equation \eqref{chareq2} admitting two pairs of complex-conjugate double roots, we can confirm that the characteristic polynomial \eqref{chareq1} has two pairs of conjugate double roots. For this class of high-order vector RW solutions, we find that the RW patterns simultaneously correspond to the root structures of two distinct Adler--Moser polynomials. However, when analyzing the asymptotics of the vector RW patterns near the locations corresponding to the polynomial roots, the computation of the leading-order coefficients with respect to the large parameters becomes highly involved. Therefore, in this work, we focus on the case where the associated Adler--Moser polynomial admits only simple roots. In this case, the leading-order term can be reduced to a sum of four terms, each with constant coefficients expressed in special determinant forms. In contrast, when the Adler--Moser polynomial has multiple roots, or when two distinct Adler--Moser polynomials share identical roots in the $ (x,t) $-plane after transformation, the computation of the leading-order terms becomes extremely complicated. As a result, the asymptotic analysis of rogue wave patterns near these locations remains to be investigated in future work.

The results shows that when the two corresponding Adler--Moser polynomials both admit only simple roots and share no common roots in the transformed $ (x,t) $-plane, each component of the resulting vector RW pattern consists of two collections of spatially separated first-order RWs. Each first-order RW can correspond to a simple root of the underlying polynomial, while those linked to different polynomials exhibit distinct dynamics---one collection displays bright structures, whereas the other exhibits four-petaled or dark structures. In each RW pattern, the two different sets of first-order RWs can be shifted to arbitrary locations in the $ (x,t) $-plane by appropriately adjusting certain free parameters. This tunability, combined with the rich simple root structures of the Adler--Moser polynomials, facilitates the systematic construction of a substantially broader class of structured RW patterns for the CNLS equation.

The structure of this paper is organized as follows: In Section \ref{Sec-RWs}, we present the explicit formula for high-order RW solutions of the CNLS equation \eqref{cnlse} and introduce the Adler--Moser polynomials. In Section \ref{Sec-asymptotic}, we investigate the asymptotic behavior of the resulting RW patterns in the condition where multiple internal parameters become large, and we illustrate several representative examples together with their dynamical evolution plots. Section \ref{Sec-proof} provides detailed proofs of our main results. Finally, conclusions and further discussions are given in Section \ref{Sec-conclusion}.

\section{Preliminaries}\label{Sec-RWs}
In this section, we will present the determinant formula of the vector RW solutions for the CNLS equation \eqref{cnlse} in the case where the associated characteristic equation \eqref{chareq1} possesses two pairs of complex-conjugate double roots. Additionally, we will introduce Adler--Moser polynomials to analyze the RW patterns below.

\subsection{Vector RW solutions of the CNLS equation}

The CNLS equation \eqref{cnlse} admits the Lax pair \cite{ling2014high}
\begin{equation} \label{laxp}
	\begin{aligned}
		&\Phi_{x}=\mathbf{U}(\lambda;x,t)\Phi,  \quad \Phi_{t}=\mathbf{V}(\lambda;x,t)\Phi, 
	\end{aligned}
\end{equation}
where
\begin{equation}\label{laxpuv}
	\begin{aligned}
		&\mathbf{U}=\ii \left( \lambda\sigma_{3} +\mathbf{Q}\right) , \quad \mathbf{V}= \lambda\mathbf{U} -\frac{\ii}{2}\sigma_{3}(\mathbf{Q}^{2}+\ii\mathbf{Q}_{x}),
		\\
		&\sigma_{3}={\rm diag}(1,-1,-1), \quad
		\mathbf{Q}=\begin{pmatrix}
			0 & \mathbf{q}^{\dagger} \\
			\mathbf{q} & \mathbf{0}_{2\times 2} 
		\end{pmatrix},
	\end{aligned}
\end{equation}
$ \Phi=\Phi(\lambda;x,t) $ is the complex matrix spectral function, and $ \lambda \in \mathbb{C} $ is a spectral parameter. We can utilize the zero curvature equation $ \mathbf{U}_{t}-\mathbf{V}_{x} +[ \mathbf{U}, \mathbf{V} ]=0$ to derive the CNLS equation \eqref{cnlse} with the commutator $[ \mathbf{U}, \mathbf{V} ] = \mathbf{U}\mathbf{V} - \mathbf{V}\mathbf{U}$.

To generate the RW solutions, we give a plane wave seed solution $ \mathbf{q}^{[0]}=\left[q_{1}^{[0]}(x,t), q_{2}^{[0]}(x,t) \right]^{T} $ for the CNLS equation \eqref{cnlse} with
\begin{equation}\label{seeds}
	\begin{aligned}
		q_{n}^{[0]}(x,t)= a_{n}\ee^{\ii \theta_{n}},  \quad
		\theta_{n}=b_{n}x+\left( a_{1}^{2}+a_{2}^{2}-\frac{b_{n}^{2}}{2}\right)t,  \quad n=1,2,
	\end{aligned}
\end{equation}
the real parameters $ a_{n} $ and $ b_{n} $ are the amplitudes and wave numbers, respectively. Since the CNLS equation \eqref{cnlse} admits the scaling symmetry and the Galilean symmetry, we can set $ a_{1}=1,\ a_{2}\in \mathbb{R}\setminus\{0\} $, and $ b_{1}=-b_{2}>0 $ without loss of generality. Then, as shown in Ref. \cite{lin2025vector}, starting from the plane-wave solution \eqref{seeds} and applying the Darboux transformation, we obtain the following formula for the high-order solution of the CNLS equation \eqref{cnlse}:
\begin{equation}\label{qkfor1}
	\begin{aligned}
		&q_{n}^{[N]}(x,t)=q^{[0]}_{n}(x,t) \frac{\det({\mathbf{M}}^{(n)})}{\det({\mathbf{M}}^{(0)})}, \quad n=1,2, \\
	\end{aligned}
\end{equation}
where the seed solutions $ q^{[0]}_{n}(x,t) $ are given in the equation \eqref{seeds}, 
\begin{equation}\label{bmsij}
	\begin{aligned}
		&{\mathbf{M}}^{(k)} = \left( {{M}}^{(k)}_{i,j} \right)_{1\leq i,j\leq N}, \quad k=0,1,2, \\
		&{{M}}_{i,j}^{(0)}= \sum_{p,r=1}^{3} \frac{c_{i,p}^{*}c_{j,r}}{\chi_{i,p}^{*}-\chi_{j,r}} \ee^{-\ii\omega_{i,p}^{*}+\ii\omega_{j,r}},\quad 
		{M}_{i,j}^{(n)}= \sum_{p,r=1}^{3} \frac{c_{i,p}^{*}c_{j,r}}{\chi_{i,p}^{*}-\chi_{j,r}} \frac{\chi_{i,p}^{*}+b_{n}}{\chi_{j,r}+b_{n}} \ee^{-\ii\omega_{i,p}^{*}+\ii\omega_{j,r}},\\
		&\omega_{j,r}= (\chi_{j,r}-\lambda_{j})x + (\frac{\chi_{j,r}^{2}}{2} -\lambda_{j}^{2} -2)t +d_{j,r}(\lambda_{j})
	\end{aligned}
\end{equation}
$ \chi_{j,r} $ is the $ r $-th root of the characteristic equation \eqref{chareq1} at $ \lambda=\lambda_{j} $, $ d_{j,r}(\lambda) $ is also the arbitrary parameters independent of $ x $ and $ t $, and all $ c_{j,r} $'s are arbitrary constants. 

Next, we derive the vector RW solution formula of the CNLS equation \eqref{cnlse} in the case where the characteristic equation \eqref{chareq1} has double roots. To obtain concise analytic expressions for the double roots of the characteristic equation \eqref{chareq1} in the complex field, we focus on the case $a_{1}=a_{2}=1$ in the plane-wave solution \eqref{seeds} and do not elaborate on the general case $a_{1}\neq a_{2}$. Note that this restriction is made only for the purpose of simplifying the explicit root parametrization; the subsequent methodology for analyzing RW patterns of the CNLS equation extends without essential modification to $a_{1}\neq a_{2}$. A detailed treatment of the general case can be found in Ref. \cite{ZhaoLing2016}. On the other hand, algebraic equation \eqref{chareq2} can be reduced to a quartic equation. When its discriminant equals zero (i.e., when $ b_{1}=\frac{1}{2} $), the quartic admits multiple roots, which is not the case of interest here. Therefore, for $ 0<b_{1}<\tfrac{1}{2} $ or \(b_{1}>\tfrac{1}{2}\), one can verify that the quartic equation has two pairs of complex-conjugate simple roots. This implies that the characteristic equation \eqref{chareq1} possesses two pairs of complex-conjugate double roots. In particular, when setting $ b_{1}=\frac{1}{2}\sqrt{1-\zeta^{2}} $ with $ \zeta\in((-1,1) \cup \ii\mathbb{R}) \backslash \{0\} $, we can obtain two pairs of complex-conjugate double roots $ \chi_{l} $ $ (l=1,2) $ and $ \chi_{l}^{*} $ for the equation \eqref{chareq1} with $ \lambda=\lambda_{l} $ and $ \lambda=\lambda_{l}^{*} $, respectively, where
\begin{equation}\label{chin}
	\begin{aligned}
		&\chi_{1}=\frac{\ii\sqrt{(\zeta+1)(\zeta+3)}}{2}, \quad \lambda_{1}=\frac{\ii(\zeta+3)\sqrt{(\zeta+1)(\zeta+3)}}{4(\zeta+1)},\\
		&\chi_{2}=\frac{\ii\sqrt{(\zeta-1)(\zeta-3)}}{2}, \quad \lambda_{2}=\frac{\ii(\zeta-3)\sqrt{(\zeta-1)(\zeta-3)}}{4(\zeta-1)}.
	\end{aligned}
\end{equation}

Moreover, we introduce a small perturbation $ \varepsilon $ for the parameters $ \lambda_{l} $ and $ \chi_{l} $, as follows:
\begin{equation}\label{lamchi}
	\lambda_{l}(\varepsilon) =\lambda_{l} +\lambda_{l}^{[2]}\varepsilon^{2}, \quad \chi_{l}(\varepsilon)=\chi_{l}+ \sum_{i=1}^{\infty}\chi_{l}^{[i]}\varepsilon^{i}, 
\end{equation}
where $ \lambda_{l}(\varepsilon) $ and $ \chi_{l}(\varepsilon) $ satisfy the characteristic equation \eqref{chareq1} with the arbitrary nonzero complex constants $ \lambda_{l}^{[2]}$. Then, for the solution formula \eqref{qkfor1}, we consider the parameters
\begin{equation}\nonumber%\label{key}
	\lambda_{j}\to\lambda_{l}(\varepsilon), \quad \chi_{j,1}\to\chi_{l}(\varepsilon),\quad  \chi_{j,2}\to \chi_{l}(-\varepsilon),
\end{equation}
with $ l=1 $ for $ 1\leq j\leq N_{1} $ and $ l=2 $ for $ N_{1}+1\leq j\leq N_{1}+N_{2} $, and set the constants $(c_{j,1},c_{j,2}, c_{j,3})=(1,-1,0) $. Further, we can perform a
Taylor expansion around $ \varepsilon $ to generate the formula \eqref{horw} of the high-order RW solution for the CNLS equation \eqref{cnlse}. The detailed computations are omitted here, which are found in our earlier analogous work \cite{linh2024a}. 

Now, we present the determinant representation of the vector RW solutions containing multiple free large parameters in the following theorem. For convenience, we introduce the Schur polynomials $ S_{j}(\mathbf{z}) $ by the generating function
\begin{equation}\label{schur}
	\sum_{j=0}^{\infty} S_{j}(\mathbf{z})\varepsilon^{j}=\exp\left( \sum_{j=1}^{\infty} z_{j}\varepsilon^{j}\right), \quad  \mathbf{z}=(z_{1}, z_{2}, \cdots), 
\end{equation}
and $  S_{j}(\mathbf{z})=0 $ for $ j<0 $, and take other notations:
\begin{equation}\label{xpm1}
	\begin{aligned}
		&\mathbf{x}_{l}^{\pm}=(x_{l,1}^{\pm}, x_{l,2}^{\pm}, \cdots), \quad x_{l,j}^{+}=\alpha_{l,j}x+\beta_{l,j}\ii t+ d_{l,j}, \quad x^{-}_{l,j}=(x^{+}_{l,j})^{*}, \\ &\mathbf{h}_{r}^{[l,m]}=(h_{r,1}^{[l,m]},h_{r,2}^{[l,m]},\cdots), \quad l,m=1,2, \quad r=1,2,3,
	\end{aligned}
\end{equation}
where $ d_{l,j} $'s are the arbitrary complex constants, $\alpha_{l,j} $, $\beta_{l,j} $, and $ h_{r,j}^{[l,m]} $ are defined by the following expansions
\begin{equation}\label{xpm2}
	\begin{aligned}
		&\ii\left( \chi_{l}(\varepsilon)-\lambda_{l}(\varepsilon)\right)  = \sum_{j=0}^{\infty}\alpha_{l,j}\varepsilon^{j}, \quad
		 \lambda_{l}^{2}(\varepsilon) - \frac{\chi_{l}^{2}(\varepsilon)}{2} +2 = \sum_{j=0}^{\infty}\beta_{l,j}\varepsilon^{j}, \\
		& \ln\left(  \dfrac{(\chi_{l}-\chi_{m}^{*})}{(\chi_{l}(\varepsilon)-\chi_{m}^{*})}\right)  =\sum_{j=1}^{\infty} h_{1,j}^{[l,m]}\varepsilon^{j}, \quad 
		\ln\left( \dfrac{(\chi_{l}-\chi_{m}^{*})}{\chi_{l}^{[1]}\varepsilon} \dfrac{(\chi_{l}(\varepsilon) -\chi_{l})}{(\chi_{l}(\varepsilon)-\chi_{m}^{*})}\right)  =\sum_{j=1}^{\infty}h_{2,j}^{[l,m]}\varepsilon^{j},    \\
		& \ln\left( \frac{\chi_{l}+(-1)^{m-1}b_{1}}{\chi_{l}(\varepsilon)+(-1)^{m-1}b_{1}}\right) =\sum_{j=1}^{\infty}{h_{3,j}^{[l,m]}}\varepsilon^{j}.
	\end{aligned}
\end{equation}

\begin{theorem}\label{Theo1}
	Given a $ 2 $-dimensional integer vector $ [N_{1}, N_{2}] $ with $N_{1},N_{2}\geq 0$ and $ N_{1}+N_{2}=N $, the CNLS equation \eqref{cnlse} admits high-order vector RW solutions $\mathbf{q}^{[N_{1}, N_{2}]}(x,t) = \left[{q}_{1}^{[N_{1}, N_{2}]}(x,t),{q}_{2}^{[N_{1}, N_{2}]}(x,t)\right]^{T} $ with
	\begin{equation}\label{horw}
		q_{n}^{[N_{1}, N_{2}]}(x,t)=  q^{[0]}_{n}(x,t)\frac{\tau^{(n)}}{\tau^{(0)}}, \quad n=1,2,
	\end{equation}
	where $ q_{n}^{[0]}(x,t) $ is the plane wave solution given in the equation \eqref{seeds}, $ \tau^{(k)} ( k=0,1,2)$ are defined by the following $ 2\times 2 $ block matrices
	\begin{equation}\label{msmat1}
		\begin{aligned}
			&\tau^{(k)}=\begin{vmatrix} 
				\mathbf{M}^{(k,{1},{1})} &  \mathbf{M}^{(k,{1},{2})}  \\
				\mathbf{M}^{(k,{2},{1})} &  \mathbf{M}^{(k,{2},{2})} 
			\end{vmatrix}_{N\times N},  \\
			&\mathbf{M}^{(k,l,m)}= \left( \tau^{(k,l,m)}_{2i-1,2j-1} \right)_{1\leq i\leq N_{l}, 1\leq j\leq N_{m}}, \quad  l,m=1,2,
		\end{aligned}
	\end{equation}
	the elements $ \tau^{(k,l,m)}_{i,j} $ are given by
	\begin{equation}\label{tauij1}
		\begin{aligned}
			&\tau_{i,j}^{(k,l,m)}=\frac{1}{\chi_{l}^{*}-\chi_{m}} \sum_{v=0}^{\min{(i,j)}}(C_{lm}^{*}C_{ml})^{v} \, S_{i-v}\left(\mathbf{x}^{-}_{l} + (\mathbf{h}_{1}^{[l,m]})^{*} +v(\mathbf{h}_{2}^{[l,m]})^{*} -(\mathbf{h}_{3}^{[l,k]})^{*} \right) \, S_{j-v}\left(  \mathbf{x}^{+}_{m} + \mathbf{h}_{1}^{[m,l]} +v\mathbf{h}_{2}^{[m,l]} +\mathbf{h}_{3}^{[m,k]}\right), 
		\end{aligned}
	\end{equation}
	with $ \mathbf{h}_{3}^{[m,0]} $ being the zero vectors, the constants $ C_{ml}= \frac{\chi_{m}^{[1]}}{\chi_{m}-\chi_{l}^{*}} $, and $ \chi_{l} $, $ S_{j} $, $\lambda_{l}^{[1]}, \chi_{l}^{[1]}$, $ \mathbf{x}_{l}^{\pm} $, and $ \mathbf{h}_{r}^{[l,m]} $ $ (1\leq r\leq 3) $ defined by the equations \eqref{lamchi}-\eqref{xpm2}.
\end{theorem}

From Theorem \ref{Theo1}, when the integer vector $ [N_{1},N_{2}] $ is equal to $ [1,0] $ or $ [0,1] $, we can obtain two different first-order fundamental vector RW solutions:
\begin{equation}\label{qk12}
	\begin{aligned}
		&\mathbf{q}^{[1, 0]}(x,t)=\left[{q}_{1}^{[1, 0]}(x,t),\, {q}_{2}^{[1, 0]}(x,t) \right]^{T}, \quad {q}_{n}^{[1, 0]}(x,t) =\hat{q}_{n,1}(x-x_{1},t-t_{1}){q}_{n}^{[0]}(x,t),\\
		&\mathbf{q}^{[0,1]}(x,t)=\left[{q}_{1}^{[0,1]}(x,t),\, {q}_{2}^{[0,1]}(x,t) \right]^{T}, \quad {q}_{n}^{[0,1]}(x,t) =\hat{q}_{n,2}(x-x_{2},t-t_{2}){q}_{n}^{[0]}(x,t), \quad n=1,2,
	\end{aligned}
\end{equation}
for the CNLS equation \eqref{cnlse}, where $ \hat{q}_{n,l}(x-x_{l},t-t_{l}) $ $ (l=1,2) $ are defined by
\begin{equation}\label{qk1}
	\begin{aligned}
		& \hat{q}_{n,l}(x,t)= 1+ \frac{2\ii}{\left| \chi_{l}+(-1)^{n+1}b_{1}\right|^{2} } \dfrac{(\Re(\chi_{l})+(-1)^{n+1}b_{1})(x+\Re(\chi_{l})t)-\Im(\chi_{l})^{2}t +\frac{\ii}{2}}{ (x+\Re(\chi_{l})t)^{2} +\Im(\chi_{l})^{2}t^{2} +\frac{1}{4\Im(\chi_{l})^{2}}},\\
		&(x_{l},t_{l})= \left( -\frac{\Re(\chi_{l})}{\Im(\chi_{l})}\Re\left(  \frac{d_{l,1}}{\chi_{l}^{[1]}}\right)  -\Im\left(  \frac{d_{l,1}}{\chi_{l}^{[1]}}\right) -\frac{1}{2\Im(\chi_{l})} ,\, \frac{1}{\Im(\chi_{l})}\Re\left(  \frac{d_{l,1}}{\chi_{l}^{[1]}}\right) \right),
	\end{aligned}
\end{equation}
with $ \Re(\chi_{l}) $ and $ \Im(\chi_{l}) $ denoted as the real and imaginary parts of $ \chi_{l} $, respectively. When $ x\rightarrow\infty $ and $ t\rightarrow\infty $, the first-order fundamental vector RW solutions $ \mathbf{q}^{[1, 0]}(x,t) $ and $\mathbf{q}^{[0,1]}(x,t)  $ both tend to the plane wave background $ \mathbf{q}^{[0]}(x,t) $ given in the equation \eqref{seeds}. At their central positions $ (x,t)=(x_{l}, t_{l}) $, the amplitudes of the first-order RWs $ |q_{n}^{[1,0]}(x,t)|  $ and $ |q_{n}^{[0,1]}(x,t)|  $ are
\begin{equation}\label{maxamp1}
	\left|1-\frac{4\Im(\chi_{l})^2}{ |\chi_{l}+(-1)^{n+1}b_{1}|^{2}}  \right|, \quad l=1,2,
\end{equation}
respectively. 

Moreover, Refs. \cite{ZhaoLing2016,qin2023optical} show that the vector first-order RW solutions of the CNLS equation \eqref{cnlse} admit three different dynamical structures, which are determined by the ratio
\begin{equation}\label{rho1}
	\rho_{n,l}=\frac{(\Re(\chi_{l})+(-1)^{n+1}b_{1})^{2}}{(\Im(\chi_{l}))^{2}}, \quad n,l=1,2,
\end{equation}
and the first-order RW is eye-shaped (i.e. a bright RW) if $\rho_{n,l} \leq \frac{1}{3} $, four-petaled if $ \frac{1}{3}<\rho_{n,l}<3  $, and anti-eye-shaped (i.e., a dark RW) if $ \rho_{n,l}\geq 3 $. Thus, we obtain the following proposition for vector first-order RW solutions $ \mathbf{q}^{[1, 0]}(x,t) $ and $ \mathbf{q}^{[0,1]}(x,t) $, and its dynamical evolution will be illustrated in Section \ref{subsec-example}.

\begin{prop}[\cite{ZhaoLing2016}]\label{prop-rw1}
	For the two first-order vector RW solutions $ \mathbf{q}^{[1, 0]}(x,t) $ and $ \mathbf{q}^{[0,1]}(x,t) $, we have
	\begin{itemize}
		\item[(1)] When $ 0<b_{1}<\frac{1}{2} $, among the two vector first-order RWs $ \mathbf{q}^{[1, 0]}(x,t) $ and $ \mathbf{q}^{[0,1]}(x,t) $, one is a four-petaled RW in each of its components, while the other is a bright RW.
		
		\item[(2)] When $ b_{1}>\frac{1}{2} $, $ {q}_{1}^{[1, 0]}(x,t) $ and $ {q}_{2}^{[0,1]}(x,t) $ are necessarily, simultaneously, dark RWs or four-petaled RWs, whereas $ {q}_{2}^{[1, 0]}(x,t) $ and $ {q}_{1}^{[0,1]}(x,t) $ are bright RWs. 
	\end{itemize}
\end{prop}
\begin{proof}
	(1) When $ 0<b_{1}<\frac{1}{2} $, we have $ \zeta\in(-1,0)\cup(0,1) $ 
	because $ b_{1}=\frac{1}{2}\sqrt{1-\zeta^{2}} $.
	Then, based on the values of $ \chi_{l} $ $ (l=1,2) $ in the equation \eqref{chin}, we derive
	$ \rho_{n,1}=\frac{1-\zeta}{3+\zeta} $ and $ \rho_{n,2}=\frac{1+\zeta}{3-\zeta} $. Thus, we deduce that $ \frac{1}{3}<\rho_{n,1}<3 $ and $ \rho_{n,2}<\frac{1}{3} $ if $ \zeta\in(-1,0) $, and $ \rho_{n,1}<\frac{1}{3} $ and $ \frac{1}{3}<\rho_{n,2}<3 $ if $ \zeta\in(0,1) $. This implies that when $ \zeta\in(-1,0) $, each component of the vector first-order RW $ \mathbf{q}^{[1, 0]}(x,t) $ is a four-petaled RW and each component of $ \mathbf{q}^{[0,1]}(x,t) $ is a bright RW; by contrast, when $ \zeta\in(0,1) $, each component of $ \mathbf{q}^{[1, 0]}(x,t) $ becomes a bright RW, while each component of $ \mathbf{q}^{[0,1]}(x,t) $ becomes a four-petaled RW.
	
	(2) When $ b_{1}>\frac{1}{2} $, we have $ \zeta\in \ii\mathbb{R} \backslash \{0\} $ 
	because $ b_{1}=\frac{1}{2}\sqrt{1-\zeta^{2}} $. Then, let
	\[\zeta=\frac{1}{2}(s+\frac{1}{s})-2, \quad s=\frac{4\bar{\zeta}}{1+\bar{\zeta}}+\ii\frac{\sqrt{\bar{\zeta}(\bar{\zeta}^{2}-14\bar{\zeta} +1)}}{1+\bar{\zeta}},\]
	with $ \bar{\zeta}\in (0,7-4\sqrt{3})\cup(7+4\sqrt{3}, \infty) $, we obtain $ b_{1}=\frac{\bar{\zeta}^{2}-6\bar{\zeta} +1}{4(1+\bar{\zeta})\sqrt{\bar{\zeta}}} $
	and further choose
	\[\chi_{1}= \frac{\sqrt{\bar{\zeta}^{2}-14\bar{\zeta} +1}}{4\sqrt{\bar{\zeta}}} +\ii\frac{1-\bar{\zeta}}{1+\bar{\zeta}}, \quad \chi_{2}=-\chi_{1}.\]
	Therefore, we can determine that $ \rho_{1,1}=\rho_{2,2}>\frac{1}{3} $ and $ \rho_{2,1}=\rho_{1,2}<\frac{1}{3} $, which implies $ {q}_{1}^{[1,0]}(x,t) $ and $ {q}_{2}^{[0,1]}(x,t) $ are four-petaled RWs or dark RWs, and $ {q}_{2}^{[1,0]}(x,t) $ and $ {q}_{1}^{[0,1]}(x,t) $ are bright RWs.
	
	This completes the proof of Proposition \ref{prop-rw1}.
\end{proof}

\subsection{Adler--Moser polynomials}\label{subsec_amp}

In 1978, Adler and Moser \cite{am1978} constructed what are now known as the Adler--Moser polynomials to obtain rational solutions of the Korteweg-de Vries (KdV) equation. The Adler--Moser polynomials $ \Theta_{N}(z,\bm{\kappa}) $ can be expressed as the following determinant \cite{am2009}:
\begin{equation}\label{amp}
	\Theta_{N}(z,\bm{\kappa})= c_{N} \det_{1\leq i,j\leq N}\left( \varphi_{2i-j}(z, \bm{\kappa})\right), 
\end{equation}
where $ c_{N}=\prod_{k=1}^{N}(2k-1)!! $, $ \bm{\kappa}=(\kappa_{3}, \kappa_{5}, \kappa_{7},\ldots) $, all $ \varphi_{k}(z, \bm{\kappa}) $'s are the special Schur polynomials defined by
\begin{equation}\label{schur2}
	\sum_{k=0}^{\infty} \varphi_{k}(z, \bm{\kappa})\epsilon^{k}=\exp\left( z\epsilon+ \sum_{j=1}^{\infty} \kappa_{2j+1}\epsilon^{2j+1}\right),
\end{equation}
$ \varphi_{k}(z, \bm{\kappa})=0 $ for $ k<0 $, $ \varphi_{k+1}'(z, \bm{\kappa})=\varphi_{k}(z, \bm{\kappa}) $, and $ \kappa_{2j+1} $ $ (j\geq 1) $ are arbitrary complex constants. Since $  \varphi_{k}(z, \bm{\kappa}) $ is a $k$th-order polynomial, we can determine that the degree of the polynomial $ \Theta_{N}(z, \bm{\kappa}) $ is $ \frac{N(N+1)}{2} $. In particular, if there is only one nonzero complex parameter $ \kappa_{2k+1}=-\frac{2^{2k}}{2k+1} $ in $ \Theta_{N}(z, \bm{\kappa}) $, the Adler--Moser polynomials can be reduced to the Yablonskii--Vorob'ev polynomial hierarchy \cite{yv2003}. Here, we present the first few Adler--Moser polynomials, as follows:
\begin{equation}\label{ampf}
	\begin{aligned}
		&\Theta_{1}(z)= z, \quad \Theta_2(z;\kappa_{3}) = z^3 - 3\kappa_3, \quad \Theta_3(z;\kappa_{3},\kappa_{5}) = z^6 - 15\kappa_3 z^3 + 45\kappa_5 z - 45\kappa_3^2,\\
		&\Theta_4(z;\kappa_{3},\kappa_{5},\kappa_{7}) = z^{10} - 45\kappa_3 z^7 + 315\kappa_5 z^5 - 1575\kappa_7 z^3 + 4725\kappa_3\kappa_5 z^2 - 4725\kappa_3^3 z - 4725\kappa_5^2+ 4725\kappa_3\kappa_7.
	\end{aligned}
\end{equation}

%\begin{equation}\label{yvpk}
%	Q_{N}^{[k]}(z)=c_{N} \det_{1\leq i,j\leq N}\left( p_{2i-j}^{[k]}(z) \right), \quad  p_{k}^{[k]}(z)= \varphi_{k}(z,\kappa_{2k+1}),
%\end{equation}
%where $ \varphi_{j}(z,\kappa_{2k+1}) $ represents $ \varphi_{j}(z,\bm{\kappa}) $ with only one parameter $ \kappa_{2k+1} $. 

Root structures of the Adler--Moser polynomials have been studied in Refs. \cite{yang2024, linh2024b, yv2003, yang2021}. Depending on the choice of free parameter vector $ \bm{\kappa} $, the Adler--Moser polynomials $ \Theta_{N}(z, \bm{\kappa}) $ exhibit a rich variety of root structures. However, these root structures share a universal feature \cite{yang2024}: the multiplicity of each multiple root must be a triangular number of the form $ \frac{N_{0}(N_{0}+1)}{2} $ with $ 1\leq N_{0}\leq N $. The multiplicity of the multiple root is equal to $ \frac{N(N+1)}{2} $ if and only if $ \bm{\kappa} $ is a zero vector. In this paper, we focus on the case where all roots of $ \Theta_{N}(z, \bm{\kappa}) $ with multiple nonzero parameters $ \kappa_{2j+1} $ are single.
Here, we present several examples for $ \Theta_{N}(z, \bm{\kappa}) $ with $ N=3,4 $, and $ \kappa_{2j+1}=\frac{5}{2j+1} $ or $ \kappa_{2j+1}=\ii $, as shown in Fig. \ref{Fig1}. 

\begin{figure*}[!htbp]
	\centering
	\includegraphics[width=\textwidth]{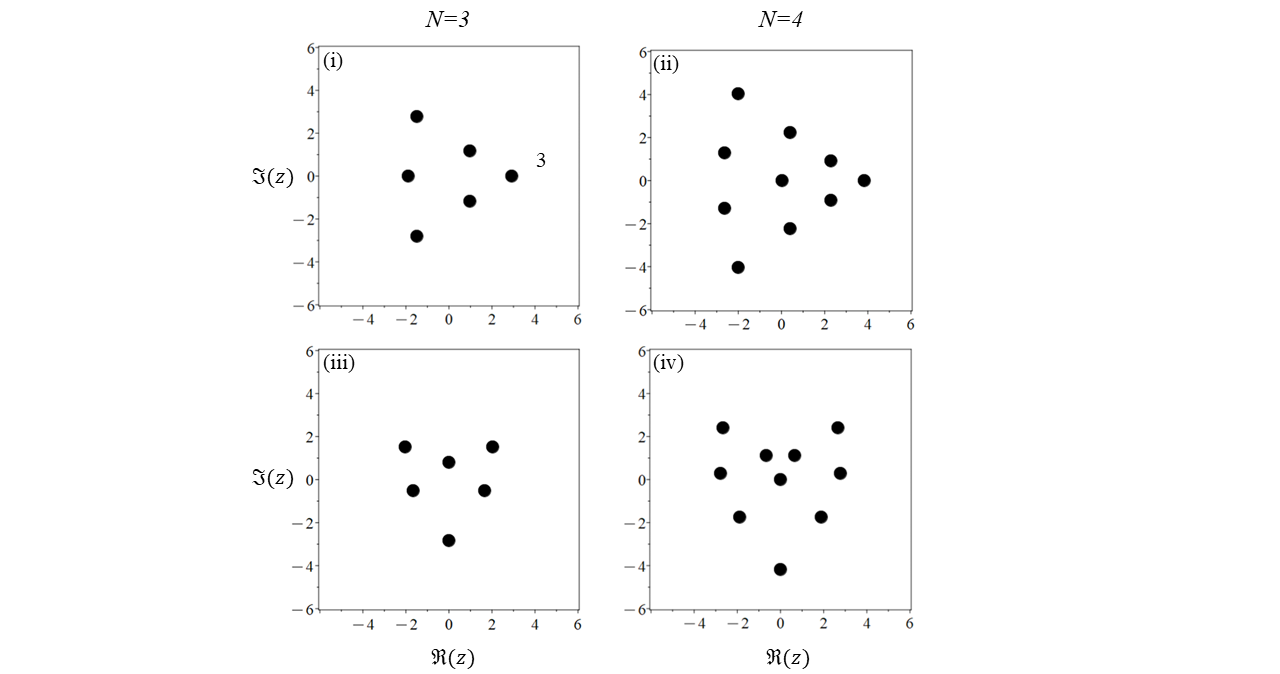}
	\caption{
		Root structures of the Adler--Moser polynomials $\Theta_{N}(z,\bm{\kappa})$, where $ N=3 $ in (i) and (iii), $ N=4 $ in (ii) and (iv), and the free parameters: $\kappa_{2j+1}=\frac{5}{2j+1}$ in (i)-(ii) and $ \kappa_{2j+1}=\ii $ in (iii)-(iv), respectively.
		Every black point represents one simple root.
	}
	\label{Fig1}
\end{figure*}

\section{asymptotic analysis and examples for the RW patterns } \label{Sec-asymptotic}

In the above Theorem \ref{Theo1}, we derived the explicit formula \eqref{horw} for high-order RW solutions $\mathbf{q}^{[N_{1}, N_{2}]}(x,t) $ of the CNLS equation \eqref{cnlse} in the case where the associated characteristic polynomial \eqref{chareq1} possesses double roots. These solutions involve multiple arbitrary complex parameters $ \{d_{l,j}\}_{1\leq l\leq 2, j\geq 1} $. When exactly one entry of $ [N_{1},N_{2}] $ is a nonzero positive integer and, among the parameters $ \{d_{l,j}\}_{1\leq l\leq 2, j\geq 1} $, only a single odd-indexed parameter $ d_{l,2j+1} $ is taken to be large, the corresponding high-order RW solutions exhibit RW patterns associated with the root structure of the Yablonskii--Vorob'ev polynomial hierarchy \cite{YangB2024book}. In this section, we investigate patterns of the high-order RW solution $\mathbf{q}^{[N_{1}, N_{2}]}(x,t) $ for the CNLS equation \eqref{cnlse} in the presence of multiple arbitrarily large parameters $ d_{l,j} $. Since the parameters $ d_{l,2j} $ $ (l=1,2, j\geq 1) $ do not affect the dynamical behavior of the RW patterns (see \cite{YangB2024book,linh2024b} for the reason), we can set $ d_{l,2j}=0 $ without loss of generality. Now, the number of the internal parameters $ d_{l,j} $ is $ N_{1}+N_{2} $. Then, following the approach in Ref. \cite{yang2024}, we choose appropriate form of the internal parameters $\{d_{l,2j+1}\}_{1\leq l \leq 2, 0 \leq j\leq N_{l}-1}$ with large value. This procedure yields new RW patterns in the CNLS equation \eqref{cnlse}, which are associated with the root structures of two distinct Adler–Moser polynomials.

Furthermore, to facilitate a rigorous analysis of the asymptotic behavior of RW patterns below, we present two lemmas and one proposition here.
\begin{lemma}\label{prop1}
	Let $N_{1}$ and $N_{2}$ be arbitrary positive integers and let $\gamma\in\mathbb{R}\setminus\{0\}$. Then the following determinant identity holds:
	\begin{equation}\label{pl-deo}
		\left| \mathbb{I}_{N_{2}} + (\gamma-1)\hat{\mathbf{R}}^{T} \hat{\mathbf{R}} \right|
		= \gamma^{N_{1}N_{2}},
	\end{equation}
	where
	\begin{equation}\label{pl-deo-w}
		\hat{\mathbf{R}} = \bigl( \hat{R}_{i,j}\bigr)_{1\leq i\leq N_{1},\,1\leq j\leq N_{2}},\qquad 
		\hat{R}_{i,j}= \gamma^{(i+j)/2 -1}  \sum_{k=1}^{\min(i,j)} \binom{i-1}{k-1}\binom{j-1}{k-1} \left( 1-\frac{1}{\gamma}\right)^{k-1}.
	\end{equation}
\end{lemma}
Its proof is provided in Sec. \ref{subsec_proofp1}.

\begin{lemma}\label{prop2}
	Let $N_{1}$ and $N_{2}$ be arbitrary positive integers. Then the identities
	\begin{equation}\label{prop-hpli}
		\begin{aligned}
			\hat{\mathbf{P}}_{12}\hat{\mathbf{P}}_{11}^{-1}
			&= \bigl( \binom{j-1}{i-1}(C_{11}-C_{12})^{\,j-i}\bigr)_{1\leq i\leq N_{12},\,1\leq j\leq N_{1}},\\[2mm]
			\hat{\mathbf{P}}_{21}\hat{\mathbf{P}}_{22}^{-1}
			&= \bigl( \binom{j-1}{i-1}(C_{22}-C_{21})^{\,j-i}\bigr)_{1\leq i\leq N_{12},\,1\leq j\leq N_{2}},
		\end{aligned}
	\end{equation}
	hold, where $ \binom{j}{i} $ is the binomial coefficient,
	\begin{equation}\label{utm}
		\begin{aligned}
			&\hat{\mathbf{P}}_{lm}
			= \bigl[ S_{j-i}(\mathbf{h}_{lm}(i-1)) \bigr]_{1\leq i\leq N_{lm},\,1\leq j\leq N_{l}}, \quad N_{lm}=\min(N_{l},N_{m}), \\
			&\mathbf{h}_{lm}(i-1)
			= \mathbf{h}_{1}^{[l,m]} + (i-1)\mathbf{h}_{2}^{[l,m]},\quad l,m=1,2,
		\end{aligned}
	\end{equation}
	and the vectors $\mathbf{h}_{n}^{[l,m]}$ $(n=1,2)$ and the constants $C_{lm}$ are given in the equations \eqref{xpm1}, \eqref{xpm2}, and \eqref{tauij1}.
\end{lemma}
See Sec. \ref{subsec_proofp2} for the proof of Lemma \ref{prop2}.

\begin{prop}\label{prop3}
	Let $N_{1}$ and $N_{2}$ be arbitrary positive integers, and let 
	$\chi_{1},\chi_{2}\in\mathbb{C}\setminus\mathbb{R}$ with $\chi_{1}\neq\chi_{2}$. 
	Then the following determinant identities hold:
	\begin{equation}\label{dete-p}
		\frac{\bigl|\mathbf{P}_{N_{1}+1,N_{2}}\bigr|\,\bigl|\mathbf{P}_{N_{1}-1,N_{2}}\bigr|}
		{\bigl|\mathbf{P}_{N_{1},N_{2}}\bigr|^{2}}
		= |C_{11}|^{2},  
		\qquad
		\frac{\bigl|\mathbf{P}_{N_{1},N_{2}+1}\bigr|\,\bigl|\mathbf{P}_{N_{1},N_{2}-1}\bigr|}
		{\bigl|\mathbf{P}_{N_{1},N_{2}}\bigr|^{2}}
		= |C_{22}|^{2},
	\end{equation}
	where
	\begin{equation}\label{dete-para2}
		\begin{aligned}
			\mathbf{P}_{N_{1},N_{2}}
			&=
			\begin{pmatrix}
				\dfrac{1}{\chi_{1}^{*}-\chi_{1}}\,\bar{\mathbf{P}}_{11}^{\dagger}\bar{\mathbf{P}}_{11}
				&
				\dfrac{1}{\chi_{1}^{*}-\chi_{2}}\,\bar{\mathbf{P}}_{12}^{\dagger}\bar{\mathbf{P}}_{21}
				\\[3mm]
				\dfrac{1}{\chi_{2}^{*}-\chi_{1}}\,\bar{\mathbf{P}}_{21}^{\dagger}\bar{\mathbf{P}}_{12}
				&
				\dfrac{1}{\chi_{2}^{*}-\chi_{2}}\,\bar{\mathbf{P}}_{22}^{\dagger}\bar{\mathbf{P}}_{22}
			\end{pmatrix},\\[2mm]
			\bar{\mathbf{P}}_{lm}
			&= \mathbf{G}_{lm}\hat{\mathbf{P}}_{lm},\qquad 
			\mathbf{G}_{lm}=\diag\bigl(1, C_{lm}, C_{lm}^{2}, \ldots, C_{lm}^{N_{lm}-1}\bigr),
		\end{aligned}
	\end{equation}
	and the matrices $\hat{\mathbf{P}}_{lm}$ are given in the equation \eqref{utm}.
\end{prop}

The proof of Proposition \ref{prop3} is given in Sec. \ref{subsec_proofp3}.

Next, let the Adler--Moser polynomials $ \Theta_{N_{l}}(z_{l}, \bm{\kappa}_{l}) $ both have only simple roots $  z_{l,i} $, and denote their respective sets of simple roots as follows:
\begin{equation}\label{srset}
	{Z}_l=\left\lbrace  z_{l,i}\in\mathbb{C}\ \mid \ \Theta_{N_l}(z_{l,i},\bm{\kappa}_l)=0,\ \Theta_{N_l}'(z_{l,i},\bm{\kappa}_l)\ne 0,\
	i=1,2,\ldots,\frac{N_l(N_l+1)}{2} \right\rbrace ,\quad l=1,2.
\end{equation}
Then let the internal large parameters $\{d_{l,2j+1}\}_{1\leq l \leq 2, 0 \leq j\leq N_{l}-1}$ in the $[N_{1}, N_{2}]$-order vector RW solution $\mathbf{q}^{[N_{1}, N_{2}]}(x,t) $  of the CNLS equation \eqref{cnlse} be the following form
\begin{equation}\label{dj1}
	d_{l,2j+1}=\kappa_{l,2j+1}A^{2j+1}, 
\end{equation}
where $A\gg 1 $ is a large positive constant, $\{\kappa_{l,2j+1}\}_{1\leq l \leq 2, 0 \leq j\leq N_{l}-1}$ are free complex constants (not all of which are zero). 
Meanwhile, for $ l=1,2 $, these roots $ z_{l,i} $ of $ \Theta_{N_{l}}(z_{l}, \bm{\kappa}_{l}) $ are in one-to-one correspondence with two set of points in the $ (x,t) $-plane,
\begin{equation}\label{srset2}
	{B}_{l} =\left\lbrace  (x_{l,i}, t_{l,i}) \in \mathbb{R}^2 \mid i=1,2,\ldots,\frac{N_{l}(N_{l}+1)}{2} \right\rbrace, \quad l=1,2,
\end{equation}
respectively, via the map
\begin{equation}\label{map0}
	\psi_{l}:\, \mathbb{R}^{2}\to \mathbb{C}, \quad \psi_{l}=\kappa_{l,1}+ \ii\chi_{l}^{[1]}(x+\chi_l t)A^{-1}, \quad l=1,2,
\end{equation}
where $ \chi_{l}^{[1]}\in \mathbb{C} $ are the arbitrary parameters, and $ \chi_{l} $ are given in equation \eqref{chin}. More precisely, for each $ i=1,2,\ldots$, $\frac{N_{l}(N_{l}+1)}{2} $, we associate $ (x_{l,i}, t_{l,i})\in B_{l} $ with each root $ z_{l,i}\in Z_{l} $ by
\[ z_{l,i}=\psi_{l}(x_{l,i}, t_{l,i}), \quad l=1,2.\]
Throughout, we assume that the above correspondence is bijective (equivalently, $ \psi_{l} $ is injective on the domain under consideration), so that each simple root $ z_{l,i} $ corresponds to a unique point $ (x_{l,i},t_{l,i}) $ in the $ (x,t) $-plane. 
Then, we can obtain the RW patterns corresponding to the root structures of the two polynomials $ \Theta_{N_{l}}(z_{l}, \bm{\kappa}_{l}) $. The detailed asymptotic analysis is presented in the following theorem.

\begin{theorem}\label{theo-nmrp}
	Let the Adler--Moser polynomials $\Theta_{N_{l}}(z_{l}, \bm{\kappa}_{l})$ possess only simple roots $ z_{l,i} \in Z_{l}$ defined in the equation \eqref{srset} with $ l=1,2 $. Assume further that the associated point sets $ B_{l} $ satisfy a uniform separation condition:
	there exists a constant $ c>0 $ such that 
	\[\|(x_{1,i},t_{1,i})-(x_{2,j},t_{2,j})\|_{2}\geq c, \quad  1\leq i\leq \frac{N_{1}(N_{1}+1)}{2},\quad  1\leq j\leq \frac{N_{2}(N_{2}+1)}{2}, \]
	where $ (x_{l,i},t_{l,i}) \in B_{l} $, and $ \|\cdot\|_{2} $ denotes the Euclidean norm on $ \mathbb{R}^{2} $.
	Then, when the internal large parameters $\{d_{l,2j+1}\}_{1\leq l \leq 2, 0 \leq j\leq N_{l}-1}$ are defined by the equation \eqref{dj1}, the pattern of the vector RW  solution $\mathbf{q}^{[N_{1}, N_{2}]}(x,t) $ for the CNLS equation \eqref{cnlse} admits the following asymptotics: 
	
	\begin{itemize}
		\item[(1).] The patterns of $ {q}_{n}^{[N_{1}, N_{2}]}(x,t) $ $ (n=1,2) $ both asymptotically split into $\frac{N_{1}(N_{1}+1)}{2} $ fundamental first-order RWs $ \hat{q}_{n,1}(x-x_{1,0}, t-t_{1,0})q_{n}^{[0]}(x,t) $ and $\frac{N_{2}(N_{2}+1)}{2} $ fundamental first-order RWs $ \hat{q}_{n,2}(x-x_{2,0}, t-t_{2,0})q_{n}^{[0]}(x,t) $, where $ q_{n}^{[0]}(x,t)  $ and $ \hat{q}_{n,l}(x,t) $ are given in equations \eqref{seeds} and \eqref{qk1}, the positions $ (x_{l,0},t_{l,0}) $ are defined by
		\begin{equation}\label{xtli1}
			\ii\chi_{l}^{[1]}\left( x_{l,0} +\chi_{l} t_{l,0}\right) +\Delta_{l}=(z_{l,0}-\kappa_{l,1})A, 
		\end{equation}
		with $ z_{l,0} $ being each of the simple roots of $ \Theta_{N_{l}}(z_{l},\bm{\kappa}_{l}) $ and the constant term $ \Delta_{l}$ given in the equations \eqref{delta1} and \eqref{delta2}.
		
		\item[(2).] When $A\gg 1$ and $\sqrt{(x-{x}_{l,0})^{2}+(t-{t}_{l,0})^{2}} = \cO(1)$, the RW solutions $ {q}_{n}^{[N_{1}, N_{2}]}(x,t) $ $ (n=1,2) $ satisfy the following asymptotic expression:
		\begin{equation}\label{asym_nmrp}
			q_{n}^{[N_{1}, N_{2}]}(x,t) =\hat{q}_{n,l}(x-{x}_{l,0}, t-{t}_{l,0}) q_{n}^{[0]}(x,t) +\cO(A^{-1}).
		\end{equation}
	
		\item[(3).] When $ A\rightarrow \infty $ and $ (x,t) $ is not near the positions of {all $ (x_{l,0},t_{l,0}) $ given in the equation \eqref{xtli1}}, the vector RW solution $\mathbf{q}^{[N_{1}, N_{2}]}(x,t)$ asymptotically approaches the vector plane wave background $\mathbf{q}^{[0]}(x,t)$ in the equation \eqref{seeds}.
	\end{itemize}

\end{theorem}

The detailed proof of Theorem \ref{theo-nmrp} is presented in Sec. \ref{subsec_proofnmrp}.

Note that when the large parameter $ A\rightarrow \infty $, the influence of the constant term \( \Delta_{l}(\bar{z}_{l,i})\) on the position of each first-order RW in the RW patterns becomes negligible. Moreover, according to the asymptotics of the vector RW solution $\mathbf{q}^{[N_{1}, N_{2}]}(x,t) $ in Theorem \ref{theo-nmrp} and the dynamics of the first-order vector RW $\mathbf{q}^{[1, 0]}(x,t) $ and $ \mathbf{q}^{[0,1]}(x,t) $ in Proposition \ref{prop-rw1}, we can obtain the following corollaries.
\begin{corol}\label{prop-asyme}
	For the high-order vector RW solution $\mathbf{q}^{[N_{1}, N_{2}]}(x,t) $ in Theorem \ref{theo-nmrp} with the large parameter $ A\gg1 $: 
	\begin{itemize}
		\item If the parameter $ b_{1} $ in the plane wave solution $ \mathbf{q}^{[0]}(x,t) $ satisfies $ 0<b_{1}<\frac{1}{2} $, the RW patterns of $ {q}_{n}^{[N_{1}, N_{2}]}(x,t) $ $ (n=1,2) $ consists of $ \frac{N_{1}(N_{1}+1)}{2} $ bright RWs and $ \frac{N_{2}(N_{2}+1)}{2} $ four-petaled RWs;
		
		\item If the parameter $ b_{1}>\frac{1}{2} $, then the RW patterns of $ {q}_{1}^{[N_{1}, N_{2}]}(x,t) $ consists of $ \frac{N_{1}(N_{1}+1)}{2} $ four-petaled or dark RWs and $ \frac{N_{2}(N_{2}+1)}{2} $ bright RWs, while those of $ {q}_{2}^{[N_{1}, N_{2}]}(x,t) $ consists of $ \frac{N_{1}(N_{1}+1)}{2} $ bright RWs and $ \frac{N_{2}(N_{2}+1)}{2} $ four-petaled or dark RWs.
	\end{itemize}
	 
\end{corol}

\begin{corol}
	When the large parameter $ A\gg1 $, the high-order vector RW solution $\mathbf{q}^{[N_{1}, N_{2}]}(x,t) $  in Theorem \ref{theo-nmrp} admits the asymptotical expression:
	\begin{equation}\label{asym_nmrp1}
		\left|q_{n}^{[N_{1},N_{2}]}(x,t)\right| =1+ \sum_{l=1}^{2} \sum_{(x_{l,0},t_{l,0})} \left(\left|\hat{q}_{n,l}(x-{x}_{l,0}, t-{t}_{l,0})\right|-1 \right) +\cO(A^{-1}), \quad n=1,2,
	\end{equation}
	where $ \hat{q}_{n,l}(x, t) $ is given in equation \eqref{qk1}, and the points $ ({x}_{l,0},{t}_{l,0}) $ are defined by equation \eqref{xtli1} with $ z_{l,0} $ traversing $ \frac{N_{l}(N_{l}+1)}{2} $ simple roots of the polynomials $\Theta_{N_{l}}(z_{l}, \bm{\kappa}_{l})$.
\end{corol}

Furthermore, when the order index $ [N_{1},N_{2}] $ in the vector RW solution $ \mathbf{q}^{[N_{1}, N_{2}]}(x,t) $ contains only one nonzero positive integer, the RW patterns in Theorem \ref{theo-nmrp} can correspond to the root structures of a single Adler--Moser polynomial, i.e., the degree of the other polynomial is zero. Now, the RW pattern in each component of $ \mathbf{q}^{[N_{1}, N_{2}]}(x,t) $ contains only a single type of first-order RW.

\subsection{Examples of RW patterns for the CNLS equation}\label{subsec-example}

This subsection presents several representative examples of RW patterns for the CNLS equation \eqref{cnlse} in Theorem~\ref{theo-nmrp}, providing an intuitive illustration of the preceding asymptotic analysis. We consider the RW patterns for the two cases: $ 0<b_{1}<\frac{1}{2} $ and $ b_{1}>\frac{1}{2} $. In particular, to visualize more clearly these RW patterns, we choose the parameters $ (\kappa_{1,1}, \kappa_{2,1}) $ appropriately so that the two regions of the RW pattern associated with the different polynomials are completely separated. To this end, in the examples below we set $ \kappa_{1,1}=-\kappa_{2,1}=c $ or $ \kappa_{1,1}=-\kappa_{2,1}^*=c $, where $ |c|>\sup\{|z|\mid z\in Z_{l},\ l=1,2\} $ is a constant with $ Z_{l} $ defined in the equation \eqref{srset}.

\bigskip
\textbf{(1)  The case $ 0<b_{1}<\frac{1}{2} $ }

Setting $ \zeta=\frac{3}{5} $ and $ b_{1}=\frac{2}{5} $, we obtain the two pairs of complex-conjugate double roots $ \chi_{l} $ and $ \chi_{l}^{*} $ $ (l=1,2)$ given by
\begin{equation}\label{pachi1}
	\begin{aligned}
		&\chi_{1}= \frac{6\ii}{5}, \quad \chi_{2}= \frac{\sqrt{6}\ii}{5},\quad \lambda_{1}=\frac{27\ii}{20}, \quad \lambda_{2}=\frac{3\sqrt{6}\ii}{5},
	\end{aligned}
\end{equation}
for the characteristic polynomial \eqref{chareq1}.
Then, with the large internal parameter $ d_{l,1}=0 $ and the arbitrary parameters
\begin{equation}\label{lamchie1}
	\lambda_{1}^{[2]}=\frac{\ii}{2}, \quad \lambda_{2}^{[2]}=\frac{4\sqrt{6}}{3}, \quad \chi_{1}^{[1]}=\chi_{2}^{[1]}=\frac{4\ii}{3},
\end{equation}
each component of the vector RW solutions $ \mathbf{q}^{[1,0]}(x,t) $ and $ \mathbf{q}^{[0,1]}(x,t) $ is a bright RW and a four-petaled RW, respectively (see Fig. \ref{Fig-b1}).

\begin{figure*}[!htbp]
	\centering
	\includegraphics[width=0.8\textwidth]{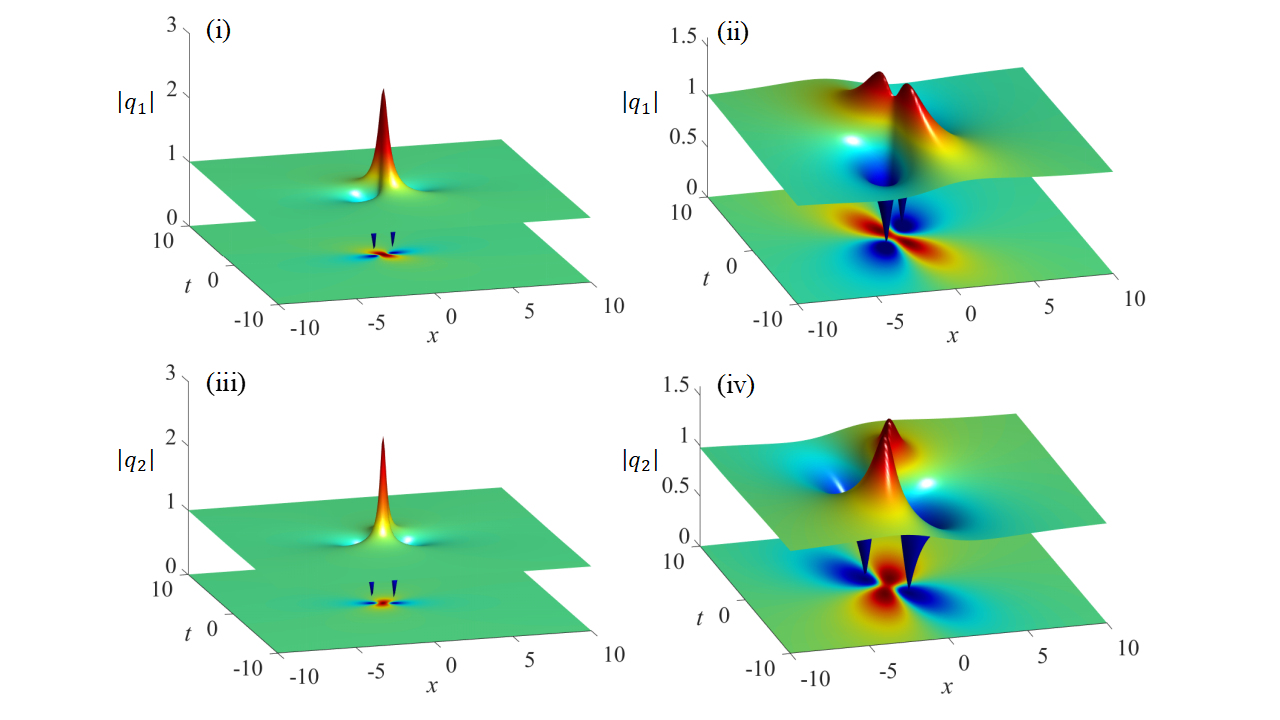}
\caption{The first-order vector RW solutions with different structures in the CNLS equation \eqref{cnlse}. The parameters are taken as $b_{1}=\frac{2}{5}$, together with $\chi_l$, $\lambda_l$, $\chi_l^{[1]}$, and $\lambda_l^{[2]}$ $ (l=1,2) $ defined in the equations \eqref{pachi1}-\eqref{lamchie1}. The first column depicts the vector bright RW solution $\mathbf{q}^{[1,0]}(x,t)$, while the second column depicts the vector four-petaled RW solution $\mathbf{q}^{[0,1]}(x,t)$. The first and second rows correspond to the first and second components of the vector RW solutions, i.e., $|q_1(x,t)|$ and $|q_2(x,t)|$, respectively.}
	\label{Fig-b1}
\end{figure*}	

Next, we consider the RW patterns of $\mathbf{q}^{[4,4]}(x,t)$ with the large internal parameters $ \{d_{l,2j+1}\}_{1\leq l\leq 2, 0\leq j\leq 3 } $ defined in the equation \eqref{dj1}. By choosing parameters $ \{\kappa_{l,2j+1}\}_{1\leq l\leq 2, 0\leq j\leq 3 } $ as one of the following sets:
\begin{equation}\label{parad11}
	\kappa_{l,1}=(-1)^{l-1}4,\quad \kappa_{l,2j+1}=\frac{5}{2j+1},
\end{equation}
\begin{equation}\label{parad12}
	\kappa_{l,1}=(-1)^{l-1}4\ii,\quad \kappa_{l,2j+1}=(-1)^{l-1}\ii,
\end{equation}
and
\begin{equation}\label{parad13}
	\kappa_{l,1}=(-1)^{l-1}4, \quad \kappa_{1,2j+1}=\frac{5}{2j+1}, \quad \kappa_{2,2j+1}=\ii,
\end{equation}
we obtain three types of vector RW patterns depicted in Fig. \ref{Fig-44RWP}, with a large constant $ A=20 $ and arbitrary parameters as defined in the equation \eqref{lamchie1}. We clearly observe that, in each component of these vector RW solutions, the RW pattern exhibits two separated regions: one consists of $10$ bright RWs and the other consists of $10$ four-petaled RWs, which correspond to all simple roots of the Adler--Moser polynomials $\Theta_{4}(z_{1}, \bm{\kappa}_{1})$ and $\Theta_{4}(z_{2}, \bm{\kappa}_{2})$, respectively. Then, we use black and yellow circles to mark the predicted locations of the bright and four-petaled RWs, respectively. Geometrically, the two regions form distinct shapes:
\begin{itemize}
	\item In Figs. \ref{Fig-44RWP}(i) and \ref{Fig-44RWP}(iv) with the free parameters \eqref{parad11}, the left and right regions form acute sectors of different sizes;
	
	\item In Figs. \ref{Fig-44RWP}(ii) and \ref{Fig-44RWP}(v) with the free parameters \eqref{parad12}, the upper and lower regions take the shape of two oppositely oriented hearts;
	
	\item In Figs. \ref{Fig-44RWP}(iii) and \ref{Fig-44RWP}(vi) with the free parameters \eqref{parad13}, the left and right regions are a heart and a sector.
\end{itemize} 

\begin{figure*}[!htbp]
	\centering
	\includegraphics[width=\textwidth]{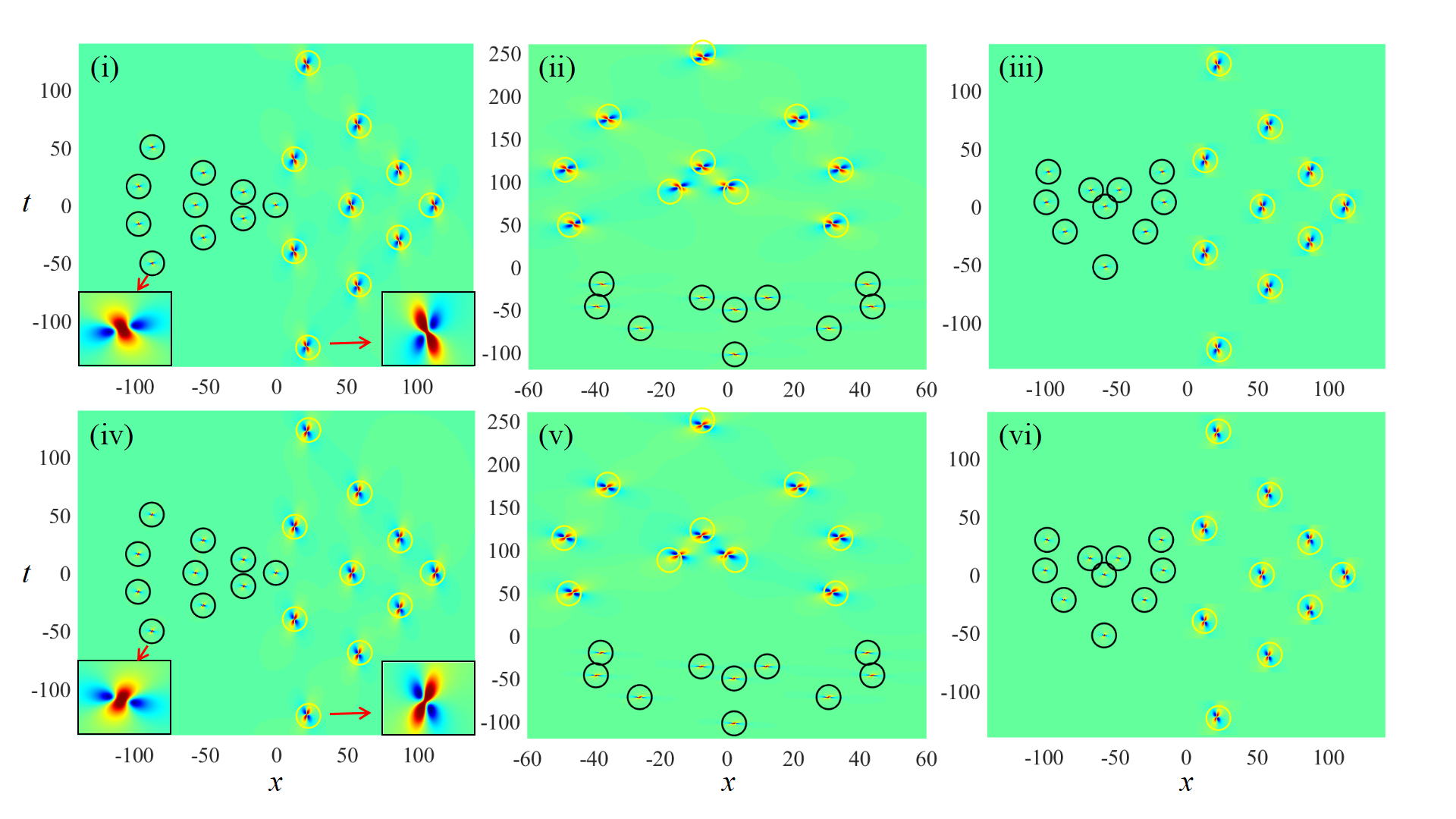}
	\caption{RW patterns of the  vector RW solution $\mathbf{q}^{[4,4]}(x,t)$ of the CNLS equation \eqref{cnlse} with the parameters $b_{1}=\frac{2}{5}$, and $\chi_l$, $\lambda_l$, $\chi_l^{[1]}$, and $\lambda_l^{[2]}$ ($l=1,2$) given in the equations \eqref{pachi1} and \eqref{lamchie1}. The large internal parameters $\{d_{l,2j+1}\}_{1\le l\le 2,\ 0\le j\le 3}$ are defined by the equation \eqref{dj1}, and the free parameters $\kappa_{l,2j+1}$ are chosen according to the equation \eqref{parad11} for panels (i) and (iv), the equation \eqref{parad12} for panels (ii) and (v), and the equation \eqref{parad13} for panels (iii) and (vi). The first and second rows correspond to the first and second components of $\mathbf{q}^{[4,4]}(x,t)$, respectively. In addition, the insets in the lower-left and lower-right corners of panels (i) and (iv) provide magnified views of a representative bright RW and a representative four-petaled RW, respectively. Black and yellow circles indicate the predicted locations of the first-order RWs in the RW patterns: black circles mark the bright RWs associated with the polynomial $\Theta_{4}(z_{1}, \bm{\kappa}_{1})$, whereas yellow circles mark the four-petaled RWs associated with the polynomial $\Theta_{4}(z_{2}, \bm{\kappa}_{2})$. }
	\label{Fig-44RWP}
\end{figure*}	

\bigskip
\noindent\textbf{(2) The case $b_{1}>\frac{1}{2}$.}

Setting $\hat{\zeta}=\frac{1}{25}$ and $b_{1}=\frac{119}{130}$, we obtain two pairs of complex-conjugate double roots $\chi_{l}$ and $\chi_{l}^{*}$ ($l=1,2$) given by
\begin{equation}\label{pachi2}
	\chi_{1}= \frac{\sqrt{69}}{10}+\frac{13\ii}{12}, \quad \chi_{2}=-\chi_{1},\quad 
	\lambda_{1}=\frac{69 \sqrt{69}}{2380}+\frac{1728 \ii}{1547}, \quad \lambda_{2}=-\lambda_{1}.
\end{equation}
Meanwhile, if $d_{l,1}=0$ and
\begin{equation}\label{lamchie2}
	\lambda_{1}^{[2]}=\frac{36000 \sqrt{69}}{1685159}+\frac{909558 \ii}{1685159}, \quad 
	\lambda_{2}^{[2]}=-\lambda_{1}^{[2]}, \quad 
	\chi_{1}^{[1]}=\chi_{2}^{[1]}=-\ii,
\end{equation}
then for the first-order vector RW solutions $\mathbf{q}^{[1,0]}(x,t)$, the first component is a dark RW and the second component is a bright RW; in contrast, for the first-order vector RW solutions $\mathbf{q}^{[0,1]}(x,t)$, the first and second components are a bright RW and a dark RW, respectively (see Fig. \ref{Fig-b2}).

\begin{figure*}[!htbp]
	\centering
	\includegraphics[width=0.8\textwidth]{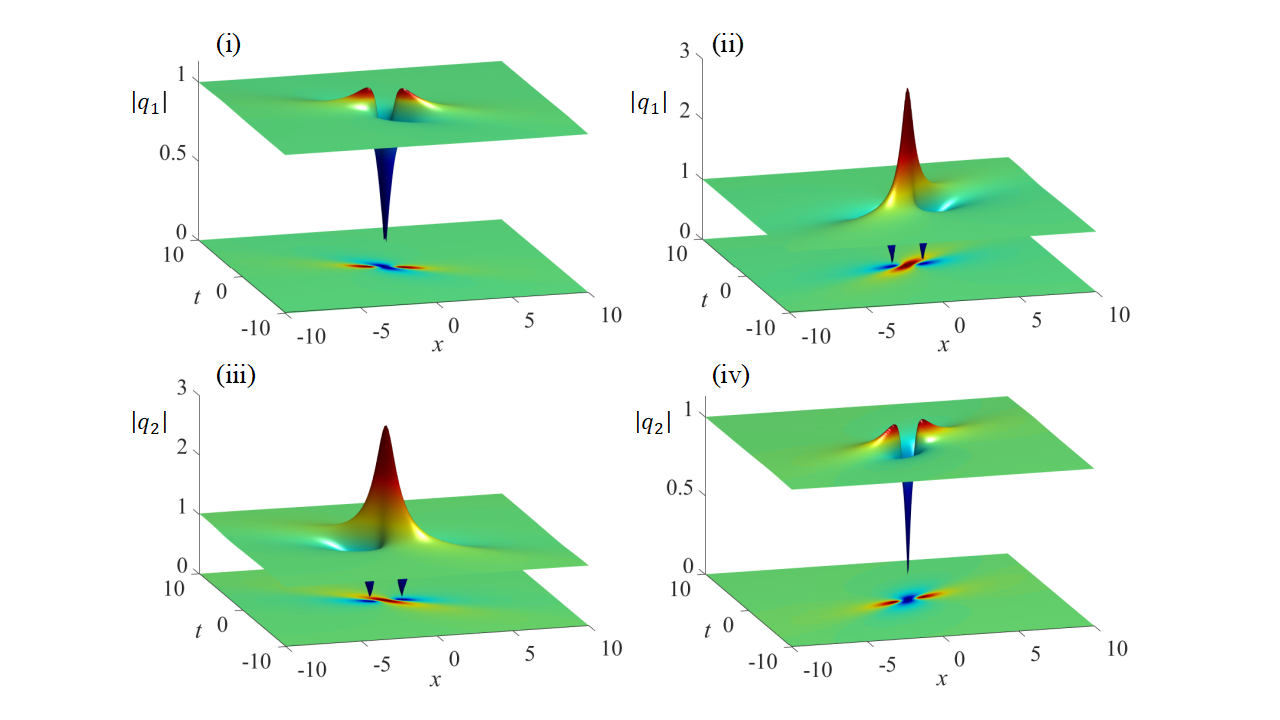}
	\caption{The first-order vector RW solutions with different structures in the CNLS equation \eqref{cnlse}. The parameters are taken as $b_{1}=\frac{119}{130}$, together with $\chi_l$, $\lambda_l$, $\chi_l^{[1]}$, and $\lambda_l^{[2]}$ $ (l=1,2) $ defined in the equations \eqref{pachi2} and \eqref{lamchie2}. Panels (i) and (iii) show the first and second components of the vector RW solution $\mathbf{q}^{[1,0]}(x,t)$, corresponding to a dark RW and a bright RW, respectively, while panels (ii) and (iv) show the first and second components of the vector RW solution $\mathbf{q}^{[0,1]}(x,t)$, corresponding to a bright RW and a dark RW, respectively.}
	\label{Fig-b2}
\end{figure*}	

Moreover, we consider patterns of the vector RW solution $\mathbf{q}^{[4,3]}(x,t)$ with $\{d_{l,2j+1}\}_{1\le l\le 2,\ 0\le j\le 3}$ defined in equation \eqref{dj1} and the parameters \eqref{pachi2}-\eqref{lamchie2}. For $A=20$, by choosing $\{\kappa_{l,2j+1}\}$ as one of the following sets:
\begin{equation}\label{parad21}
	\kappa_{l,1}=(-1)^{l-1}4,\quad \kappa_{l,2j+1}=(-1)^{l-1}\frac{5}{2j+1},
\end{equation}
\begin{equation}\label{parad22}
	\kappa_{l,1}=(-1)^{l-1}4,\quad \kappa_{l,2j+1}=\ii,
\end{equation}
and
\begin{equation}\label{parad23}
	\kappa_{1,1}=3 + 3\ii,\quad \kappa_{2,1}=-\kappa_{1,1}^{*}, \quad \kappa_{1,2j+1}=\frac{5}{2j+1}, \quad \kappa_{2,2j+1}=\ii,
\end{equation}
we produce three vector RW patterns shown in Fig. \ref{Fig-43RWP}. 
In this case, the RW patterns of $q_{1}^{[4,3]}(x,t)$ (resp.\ $q_{2}^{[4,3]}(x,t)$) again decompose into two well-separated regions, consisting of $10$ dark (resp.\ bright) first-order RWs and $6$ bright (resp.\ dark) first-order RWs, respectively. These two regions correspond to the simple-root structures of the Adler--Moser polynomials $\Theta_{4}(z_{1},\bm{\kappa}_{1})$ and $\Theta_{3}(z_{2},\bm{\kappa}_{2})$, respectively. In Fig. \ref{Fig-43RWP}, black and yellow circles mark the predicted locations of the first-order RWs associated with $\Theta_{4}(z_{1},\bm{\kappa}_{1})$ and $\Theta_{3}(z_{2},\bm{\kappa}_{2})$, respectively. In particular, the black (resp.\ yellow) circles mark dark (resp.\ bright) RWs in Fig. \ref{Fig-43RWP} (i)-(iii), whereas the correspondence is reversed in Fig. \ref{Fig-43RWP} (iv)-(vi), i.e., black (resp.\ yellow) circles mark bright (resp.\ dark) RWs.
Geometrically, the two regions form the following different shapes:
\begin{itemize}
	\item In Figs. \ref{Fig-43RWP}(i) and \ref{Fig-43RWP}(iv) with free parameters \eqref{parad21}, the left and right regions form two oppositely oriented triangles;
	
	\item In Figs. \ref{Fig-43RWP}(ii) and \ref{Fig-43RWP}(v) with free parameters \eqref{parad22}, the left and right regions take the shape of two oppositely oriented hearts;
	
	\item In Figs. \ref{Fig-43RWP}(iii) and \ref{Fig-43RWP}(vi) with free parameters \eqref{parad23}, the left-lower and right-upper regions are a triangle and a heart.
\end{itemize}

It is worth noting that, in all the above examples of high-order vector RW solutions, the RW pattern in each component simultaneously contains two distinct types of first-order RWs. For instance, the patterns in Fig. \ref{Fig-44RWP} comprise both bright and four-petaled RWs, whereas those in Fig. \ref{Fig-43RWP} comprise both dark and bright RWs. For all the examples, the predicted locations of the constituent first-order RWs are determined by Theorem \ref{theo-nmrp} and agree asymptotically with the actual RW positions; the residual discrepancy becomes negligible in the large-parameter regime $A\gg 1$. The configuration of each region can be interpreted, to some extent, as a scaled and rotated imprint of the simple-root structure of the associated Adler--Moser polynomials. Moreover, tuning the free parameters allows the two regions to be shifted to essentially arbitrary locations in the $(x,t)$-plane, which in turn enables the systematic generation of a broader family of structured RW patterns for the CNLS equation \eqref{cnlse}.

\begin{figure*}[!htbp]
	\centering
	\includegraphics[width=\textwidth]{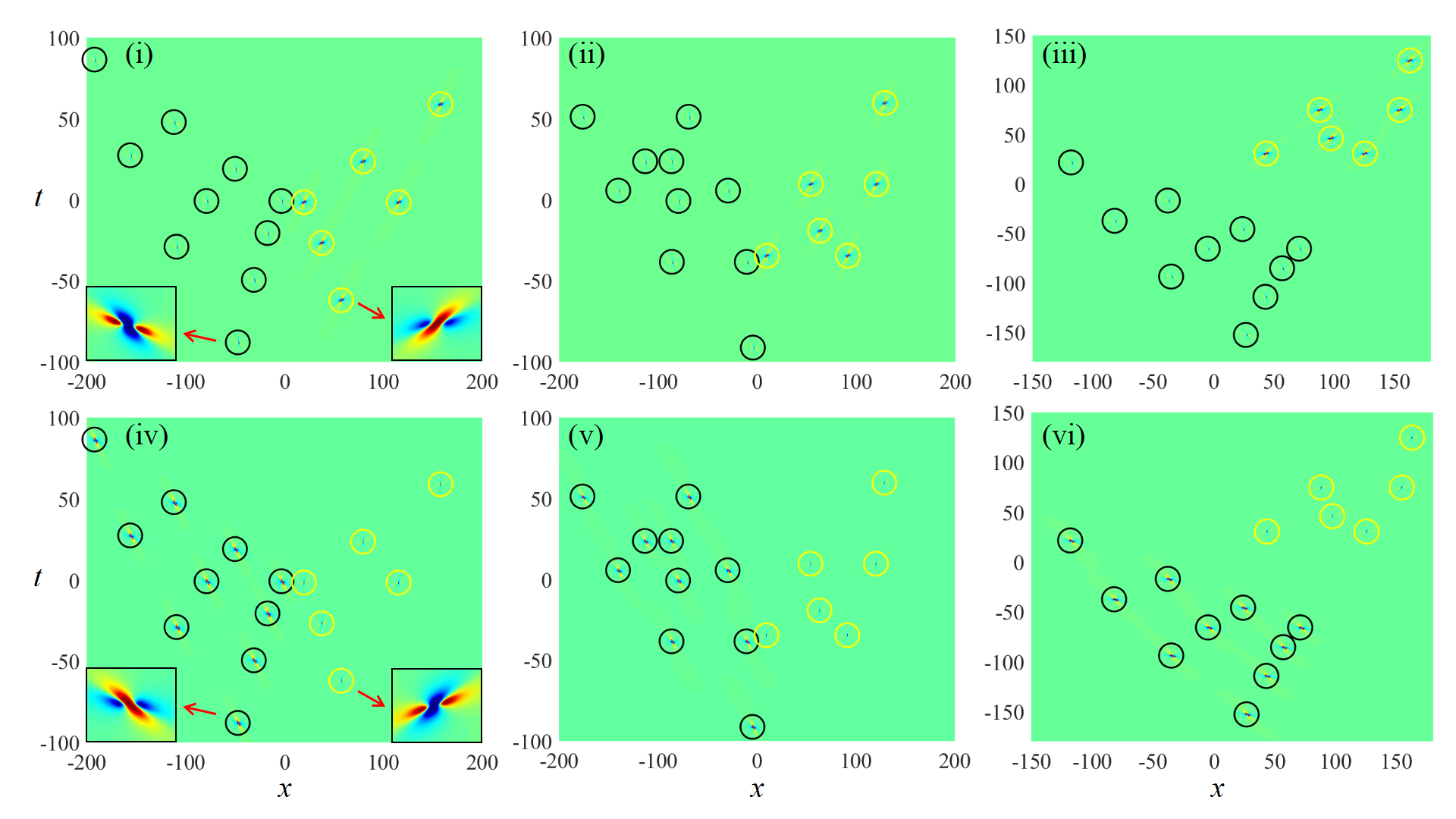}
\caption{RW patterns of the vector RW solution $\mathbf{q}^{[4,3]}(x,t)$ of the CNLS equation \eqref{cnlse} with the parameters $b_{1}=\frac{119}{130}$, and $\chi_l$, $\lambda_l$, $\chi_l^{[1]}$, and $\lambda_l^{[2]}$ ($l=1,2$) given in the equations \eqref{pachi2} and \eqref{lamchie2}. The large internal parameters $\{d_{l,2j+1}\}_{1\le l\le 2,\ 0\le j\le 3}$ are defined by the equation \eqref{dj1} and the free parameters $\kappa_{l,2j+1}$ are chosen according to the equation \eqref{parad21} for (i) and (iv), the equation \eqref{parad22} for (ii) and (v), and the equation \eqref{parad23} for (iii) and (vi). The first and second rows correspond to the first and second components of $\mathbf{q}^{[4,3]}(x,t)$, respectively. In addition, the insets in the lower corners of panels (i) and (iv) provide magnified views of a representative dark RW and a representative bright RW. In these RW patterns, black and yellow circles indicate the predicted locations of the first-order RWs associated with the polynomials $\Theta_{4}(z_{1}, \bm{\kappa}_{1})$ and $\Theta_{3}(z_{2}, \bm{\kappa}_{2})$, respectively. In particular, the black (resp.\ yellow) circles mark dark (resp.\ bright) RWs in panels (i)-(iii), while the black (resp.\ yellow) circles mark bright (resp.\ dark) RWs in panels (iv)-(vi).}
	\label{Fig-43RWP}
\end{figure*}

%%%%%%%%%%%%%%%%%%%%
\section{Proofs of the main results}\label{Sec-proof}

\subsection{Proof of Proposition \ref{prop1}}\label{subsec_proofp1}	
This subsection provides the proof of Lemma \ref{prop1}.

We introduce some notations:
\begin{equation}\label{prop1-pfno}
	\begin{aligned}
		& \mathbb{I}_{N_{2}} +(\gamma-1)\hat{\mathbf{R}}^{T} \hat{\mathbf{R}}= \left( R_{i,j}\right)_{1\leq i,j\leq N_{2}}, \\
		& \hat{\mathbf{R}}_{j}(\varepsilon)= \sum_{v=1}^{\infty} \hat{R}_{v,j} \varepsilon^{v-1}, \quad \tilde{\mathbf{R}}_{j}(\varepsilon)=\sum_{v=1}^{N_{1}} \hat{R}_{v,j} \varepsilon^{v-1}.
	\end{aligned}
\end{equation}
Then, we will first prove the equality of
\begin{equation}\label{prop-rijf}
	R_{i,j}=R_{i+1,j+1}-\frac{R_{i+1,1}}{R_{1,1}}R_{1,j+1},
\end{equation}
which is important for computation of the determinant \eqref{pl-deo}.

According to the definition of the matrix $ \hat{\mathbf{R}} $ in the equation \eqref{pl-deo-w}, we have
\begin{equation}\label{prop-rij}
	\begin{aligned}
		&R_{i,j}=\delta_{i,j} +(\gamma-1)\sum_{v=1}^{N_{1}} \hat{R}_{v,i}\hat{R}_{v,j}, \quad 
		\delta_{i,j}=\begin{cases}
			1, & i=j,\\
			0, & i\ne j,
		\end{cases}
	\end{aligned}
\end{equation}
and derive $ R_{i+1,j+1}-R_{i,j}=(\gamma-1)\left( \sum_{v=1}^{N_{1}} \hat{R}_{v,i+1}\hat{R}_{v,j+1} -\sum_{v=1}^{N_{1}} \hat{R}_{v,i}\hat{R}_{v,j}\right) $. 
Since the constant term in the polynomial $ \tilde{\mathbf{R}}_{i}(\varepsilon)\tilde{\mathbf{R}}_{j}(\varepsilon^{-1}) $ is $ \sum_{v=1}^{N_{1}} \hat{R}_{v,i}\hat{R}_{v,j} $,
it is found that $ \left( R_{i+1,j+1}-R_{i,j} \right) (\gamma-1)^{-1}$ is equal to the constant term in the polynomial 
\begin{equation}\label{prop-poly}
	\tilde{\mathbf{R}}_{i+1}(\varepsilon)\tilde{\mathbf{R}}_{j+1}(\varepsilon^{-1}) -\tilde{\mathbf{R}}_{i}(\varepsilon)\tilde{\mathbf{R}}_{j}(\varepsilon^{-1}).
\end{equation}
Thus, we will calculate the constant term of the polynomial \eqref{prop-poly} below.

Since
\begin{equation}\label{prop-hrjv}
	\begin{aligned}
		\hat{\mathbf{R}}_{j}(\varepsilon)=& \gamma^{(j-1)/2}  \sum_{v=1}^{\infty} \gamma^{(v-1)/2} \sum_{k=1}^{\min(v,j)} \binom{v-1}{k-1}\binom{j-1}{k-1} \left( 1-\frac{1}{\gamma}\right) ^{k-1}\varepsilon^{v-1}\\
		=& \gamma^{(j-1)/2}  \sum_{k=1}^{j} \binom{j-1}{k-1}\left( 1-\frac{1}{\gamma}\right) ^{k-1} \sum_{v=k}^{\infty} \binom{v-1}{k-1} \gamma^{(v-1)/2}\varepsilon^{v-1}\\
		=& \gamma^{(j-1)/2}  \sum_{k=1}^{j} \binom{j-1}{k-1}\left( 1-\frac{1}{\gamma}\right) ^{k-1} \frac{\left( \gamma^{1/2}\varepsilon\right)^{k-1} }{\left(1-\gamma^{1/2}\varepsilon \right)^{k} }\\
		=&\gamma^{(j-1)/2} \frac{\left(1-\gamma^{-1/2}\varepsilon \right)^{j-1}}{\left(1-\gamma^{1/2}\varepsilon \right)^{j}},\\
	\end{aligned}
\end{equation}
we derive
\begin{equation}\nonumber
	\hat{\mathbf{R}}_{j+1}(\varepsilon)= \gamma^{1/2} \frac{1-\gamma^{-1/2}\varepsilon }{1-\gamma^{1/2}\varepsilon } \hat{\mathbf{R}}_{j}(\varepsilon),
\end{equation}
and further obtain
\begin{equation}\label{prop-hrjf}
	\hat{\mathbf{R}}_{j+1}(\varepsilon)-\gamma^{1/2}\hat{\mathbf{R}}_{j}(\varepsilon)=\left( \gamma^{1/2}\hat{\mathbf{R}}_{j+1}(\varepsilon)-\hat{\mathbf{R}}_{j}(\varepsilon) \right) \varepsilon.
\end{equation}
%	i.e.,
%	\begin{equation}\label{prop-hrjf2}
	%		\sum_{v=1}^{\infty} \left( \hat{R}_{v,j+1} -\gamma^{1/2}\hat{R}_{v,j}\right) \varepsilon^{v}= \sum_{v=1}^{\infty} \left( \gamma^{1/2}\hat{R}_{v,j+1} -\hat{R}_{v,j}\right) \varepsilon^{v+1}.
	%	\end{equation}
We truncate the above equation \eqref{prop-hrjf}, retaining the terms in $ \varepsilon $ from $ 0 $-order up to $ N_{1} $-order, and thereby yield
\begin{equation}\nonumber
	\tilde{\mathbf{R}}_{j+1}(\varepsilon)-\gamma^{1/2}\tilde{\mathbf{R}}_{j}(\varepsilon) +\left( \hat{R}_{N_{1},j+1} -\gamma^{1/2}\hat{R}_{N_{1},j}\right) \varepsilon^{N_{1}}=\left( \gamma^{1/2}\tilde{\mathbf{R}}_{j+1}(\varepsilon)-\tilde{\mathbf{R}}_{j}(\varepsilon) \right) \varepsilon.
\end{equation} 
Thus, we have
\begin{equation}\label{prop-trj}
	\left( 1-\gamma^{1/2}\varepsilon\right) \tilde{\mathbf{R}}_{j+1}(\varepsilon) + \left( \varepsilon- \gamma^{1/2}\right) \tilde{\mathbf{R}}_{j}(\varepsilon)= s_{j}\varepsilon^{N_{1}},
	\quad s_{j}= \hat{R}_{N_{1},j+1} -\gamma^{1/2}\hat{R}_{N_{1},j},
\end{equation}
and then derive
\begin{equation}\label{prop-trj2}
	\tilde{\mathbf{R}}_{j+1}(\varepsilon) = \left( 1-\gamma^{1/2}\varepsilon\right)^{-1} \left( s_{j}\varepsilon^{N_{1}}- \left( \varepsilon- \gamma^{1/2}\right) \tilde{\mathbf{R}}_{j}(\varepsilon)\right). 
\end{equation}
By using the above formula \eqref{prop-trj2}, we yield
\begin{equation}\label{prop-trj3}
	\tilde{\mathbf{R}}_{i+1}(\varepsilon)\tilde{\mathbf{R}}_{j+1}(\varepsilon^{-1}) = \tilde{\mathbf{R}}_{i}(\varepsilon)\tilde{\mathbf{R}}_{j}(\varepsilon^{-1}) -\frac{s_{j}\varepsilon^{1-N_{1}}}{1-\gamma^{1/2}\varepsilon} \tilde{\mathbf{R}}_{i}(\varepsilon) -\frac{s_{i}\varepsilon^{N_{1}}}{\varepsilon-\gamma^{1/2}}\tilde{\mathbf{R}}_{j}(\varepsilon^{-1}) +\frac{s_{i}s_{j}\varepsilon}{(1-\gamma^{1/2}\varepsilon)(\varepsilon -\gamma^{-1/2})}.
\end{equation}
It is found that
\[ \frac{s_{i}\varepsilon^{N_{1}}}{\varepsilon-\gamma^{1/2}}\tilde{\mathbf{R}}_{j}(\varepsilon^{-1})=\cO(\varepsilon), \quad \frac{s_{i}s_{j}\varepsilon}{(1-\gamma^{1/2}\varepsilon)(\varepsilon -\gamma^{-1/2})}=\cO(\varepsilon),
\]
which implies that the constant term of the polynomial \eqref{prop-poly} is equal to the constant term of $ -\frac{s_{j}\varepsilon^{1-N_{1}}}{1-\gamma^{1/2}\varepsilon} \tilde{\mathbf{R}}_{i}(\varepsilon) $.

Since
\begin{equation}\nonumber
	-\frac{s_{j}\varepsilon^{1-N_{1}}}{1-\gamma^{1/2}\varepsilon} \tilde{\mathbf{R}}_{i}(\varepsilon)= -s_{j}\sum_{k=0}^{\infty}\gamma^{k/2}\left( \sum_{v=1}^{N_{1}} \hat{R}_{v,i}\varepsilon^{v+k-N_{1}} \right),
\end{equation}
it is computed that the constant term of $ -\frac{s_{j}\varepsilon^{1-N_{1}}}{1-\gamma^{1/2}\varepsilon} \tilde{\mathbf{R}}_{i}(\varepsilon) $ is
\begin{equation}\label{prop-cont0}
	%-s_{i}\sum_{v=1}^{N_{1}} \hat{R}_{v,i}\gamma^{(N_{1}-v)/2} = 
	-s_{j}\gamma^{(N_{1}-1)/2} \tilde{\mathbf{R}}_{i}(\gamma^{-1/2}).
\end{equation}
Additionally, when $ \varepsilon=\gamma^{-1/2} $ in the equation \eqref{prop-trj}, we obtain
\begin{equation}\label{prop-tric}
	\tilde{\mathbf{R}}_{i}(\gamma^{-1/2})= \frac{s_{i}}{1-\gamma} \gamma^{(1-N_{1})/2}.
\end{equation}
Therefore, we can determine that the constant term of $ -\frac{s_{j}\varepsilon^{1-N_{1}}}{1-\gamma^{1/2}\varepsilon} \tilde{\mathbf{R}}_{i}(\varepsilon) $, (i.e., the constant term of the polynomial \eqref{prop-poly}) is
\begin{equation}\label{prop-cont}
	\frac{s_{i}s_{j}}{\gamma-1}.
\end{equation}
On the other hand, when $ \varepsilon=\gamma^{1/2} $ in the equation \eqref{prop-trj}, we also get
\begin{equation}\label{prop-si}
	\begin{aligned}
		s_{j}=&\left( 1-\gamma\right)\gamma^{-N_{1}/2} \tilde{\mathbf{R}}_{j+1}(\gamma^{1/2})\\
		=& \left( 1-\gamma\right)\gamma^{-N_{1}/2} \sum_{k=1}^{N_{1}} \hat{R}_{k,j+1} \gamma^{(k-1)/2} \\
		=&-\gamma^{-N_{1}/2} R_{1,j+1},
	\end{aligned}
\end{equation}
because of $ \hat{R}_{k,1}=\gamma^{(k-1)/2} $. Similarly, we can get $ s_{i}=-\gamma^{-N_{1}/2} R_{i+1,1} $. Hence, as $ R_{1,1}=\gamma^{N_{1}} $, the constant term \eqref{prop-cont} of the polynomial \eqref{prop-poly} is equal to
\begin{equation}\label{prop-cont2}
	\frac{1}{\gamma-1} \frac{R_{i+1,1}}{R_{1,1}}R_{1,j+1},
\end{equation}
which proves the equation \eqref{prop-rijf}.

Then, we apply the formula \eqref{prop-rijf} to perform row operations on the determinant
\begin{equation}\label{detecom}
	\begin{aligned}
		\left| \mathbb{I}_{N_{2}} +(\gamma-1)\hat{\mathbf{R}}^{T} \hat{\mathbf{R}}\right| &= R_{1,1} \det_{1\leq i,j \leq N_{2}-1} \left(R_{i,j} \right)\\
		&= R_{1,1}^{2}\det_{1\leq i,j \leq N_{2}-2} \left(R_{i,j} \right)\\
		&\cdots\\
		&=\gamma^{N_{1}N_{2}}.
	\end{aligned}
\end{equation}

This concludes the proof of Lemma \ref{prop1}.

\subsection{Proof of Lemma \ref{prop2}}\label{subsec_proofp2}	
A proof of Lemma \ref{prop2} is given in this subsection.

First, we factor the matrix $ \hat{\mathbf{P}}_{11}^{-1} $, as follows:
\begin{equation}\label{prop-hp11i}
	\hat{\mathbf{P}}_{11}^{-1}=\mathbf{L}_{1}\mathbf{L}_{2}\cdots\mathbf{L}_{N_{1}-1},\\
\end{equation}
where
\begin{equation}\label{prop-hp11li}
	\begin{aligned}
		&\mathbf{L}_{1}=\begin{pmatrix}
			1 & -S_{1}(\mathbf{h}_{11}{(0)}) & -S_{2}(\mathbf{h}_{11}{(0)}) & \cdots & -S_{N_{1}-1}(\mathbf{h}_{11}{(0)})  \\
			0 & 1 & 0 & \cdots & 0\\
			0 & 0 & 1 & \cdots & 0 \\
			\vdots & \vdots & \vdots &\ddots  & \vdots \\
			0 & 0 & 0 & \cdots & 1 \\
		\end{pmatrix}_{N_{1}\times N_{1}}, \\
		&\mathbf{L}_{2}=\begin{pmatrix}
			1 & 0 & 0 & \cdots & 0  \\
			0 & 1 & -S_{1}(\mathbf{h}_{11}{(1)}) & \cdots & -S_{N_{1}-2}(\mathbf{h}_{11}{(1)})\\
			0 & 0 & 1 & \cdots & 0 \\
			\vdots & \vdots & \vdots &\ddots  & \vdots \\
			0 & 0 & 0 & \cdots & 1 \\
		\end{pmatrix}_{N_{1}\times N_{1}}, \\
		& \quad \vdots\\
		&\mathbf{L}_{N_{1}-1}=\begin{pmatrix}
			1 & 0 &  \cdots & 0 & 0 \\
			0 & 1 &  \cdots & 0 & 0\\
			\vdots & \vdots & \ddots &\vdots  & \vdots \\
			0 &  0 & \cdots & 1 & -S_{1}(\mathbf{h}_{11}{(N_{1}-2)}) \\
			0 &  0 & \cdots & 0 & 1 \\
		\end{pmatrix}_{N_{1}\times N_{1}}. \\
	\end{aligned}
\end{equation}
namely,
\begin{equation}\label{pr-imf}
	\left( \mathbf{L}_{k}\right)_{i,j} =\begin{cases}
		-S_{j-i}(\mathbf{h}_{11}{(i-1)}), & j>i,i=k,\\
		1, & j=i,\\
		0, & \text{others},
	\end{cases}
\end{equation}
with $ \left( \mathbf{L}_{k}\right)_{i,j} $ denoted as $ (i,j) $-entry of the matrix $ \mathbf{L}_{k} $ $ (1\leq k\leq N_{1}-1) $. Thus, $ \hat{\mathbf{P}}_{12}\hat{\mathbf{P}}_{11}^{-1} $ can be viewed as performing $ N_{1}-1 $ successive groups of column transformations $\left(  \mathbf{L}_{1},\mathbf{L}_{2},\ldots,\mathbf{L}_{N_{1}-1} \right) $ to $ \hat{\mathbf{P}}_{12} $. 

Moreover, we denote the $ (i,j) $-entry of $ \hat{\mathbf{P}}_{12} $ by $ p_{i,j}^{(0)} $, and let $ p_{i,j}^{(k)} $ denote the value of $ p_{i,j}^{(0)} $ after $ k $ column transformations, as follows:
\begin{equation}\label{prop-pk}
	\begin{aligned}
		&\hat{\mathbf{P}}_{12}\mathbf{L}_{1}=\begin{pmatrix}
			1 & p_{1,2}^{(1)} & p_{1,3}^{(1)} & \cdots & p_{1,N_{1}}^{(1)}  \\
			0 & 1 & p_{2,3}^{(0)} & \cdots & p_{2,N_{1}}^{(0)}\\
			0 & 0 & 1 & \cdots & p_{3,N_{1}}^{(0)} \\
			\vdots & \vdots & \vdots &\ddots  & \vdots \\
		\end{pmatrix}_{N_{12}\times N_{1}}, \\
		&\hat{\mathbf{P}}_{12}\mathbf{L}_{1}\mathbf{L}_{2}=\begin{pmatrix}
			1 & p_{1,2}^{(1)} & p_{1,3}^{(2)} & \cdots & p_{1,N_{1}}^{(2)}  \\
			0 & 1 & p_{2,3}^{(1)} & \cdots & p_{2,N_{1}}^{(1)}\\
			0 & 0 & 1 & \cdots & p_{3,N_{1}}^{(0)} \\
			\vdots & \vdots & \vdots &\ddots  & \vdots  \\
		\end{pmatrix}_{N_{12}\times N_{1}},\\
		& \quad \vdots \\
%		&\hat{\mathbf{P}}_{12}\mathbf{L}_{1}\mathbf{L}_{2}\cdots \mathbf{L}_{N_{1}-1}=\begin{pmatrix}
%			1 & p_{1,2}^{(1)} & p_{1,3}^{(2)} & \cdots & p_{1,N_{1}}^{(N_{1}-1)}  \\
%			0 & 1 & p_{2,3}^{(1)} & \cdots & p_{2,N_{1}}^{(N_{1}-2)}\\
%			0 & 0 & 1 & \cdots & p_{3,N_{1}}^{(N_{1}-3)} \\
%			\vdots & \vdots & \vdots &\ddots  & \vdots  \\
%		\end{pmatrix}_{N_{12}\times N_{1}},
	\end{aligned}
\end{equation}
At this point, the $ (i,j) $-entry of the matrix $ \hat{\mathbf{P}}_{12}\hat{\mathbf{P}}_{11}^{-1} $ is represented by
\begin{equation}\label{pr-imf2}
	\left( \hat{\mathbf{P}}_{12}\hat{\mathbf{P}}_{11}^{-1} \right)_{i,j} =\begin{cases}
		p_{i,j}^{(j-i)}, & i<j,\\
		1, & i=j,\\
		0, & i>j.
	\end{cases}
\end{equation}
We observe that the nonzero entry $ p_{i,i+k}^{(k)} $ of $ \hat{\mathbf{P}}_{12}\hat{\mathbf{P}}_{11}^{-1} $ is obtained from the entry $ p_{i,i+k}^{(k-1)} $ of $ \hat{\mathbf{P}}_{12}\mathbf{L}_{1} \mathbf{L}_{2} \cdots \mathbf{L}_{i+k-2} $ by performing one column operation, that is, by right-multiplying the matrix by $ \mathbf{L}_{i+k-1} $ with $ 1\leq k \leq N_{1}-i $. Thus, we have
\begin{equation}\label{prop-pkf1}
	p_{i,i+k}^{(k)}=\, p_{i,i+k}^{(k-1)} + \left( \mathbf{L}_{i+k-1}\right)_{i+k-1,i+k} p_{i,i+k-1}^{(k-1)}.
\end{equation}
Similarly, we yield
\begin{equation}\label{prop-pkf2}
		p_{i,i+k}^{(k-1)}=\, p_{i,i+k}^{(k-2)} + \left( \mathbf{L}_{i+k-2}\right)_{i+k-2,i+k} p_{i,i+k-2}^{(k-2)},\quad
		\ldots
\end{equation}
Therefore, we obtain the following expression:
\begin{equation}\label{prop-pkf}
\begin{aligned}
	p_{i,i+k}^{(k)}=\,& p_{i,i+k}^{(0)} + \sum_{v=0}^{k-1} \left( \mathbf{L}_{i+v}\right)_{i+v,i+k} p_{i,i+v}^{(v)}  \\
	=\,& S_{k}(\mathbf{h}_{12}{(i-1)}) - \sum_{v=0}^{k-1}S_{k-v}(\mathbf{h}_{11}{(i-1+v)}) p_{i,i+v}^{(v)}, \quad i=1,2,\ldots, N_{1}-k.
\end{aligned}
\end{equation}

Then, we denote
\begin{equation}\label{prop-notion}
	f_{i}(\varepsilon)=\sum_{k=0}^{\infty} p_{i,i+k}^{(k)} \varepsilon^{k}, \quad g_{n}^{[1,m]}(\varepsilon)=\exp\left(\sum_{k=1}^{\infty} h_{n,k}^{[1,m]}\varepsilon^{k} \right), \quad m,n=1,2,
\end{equation}
with $ h_{n,k}^{[1,m]} $ given in the equation \eqref{xpm2}, and thereby obtain
\begin{equation}\label{prop-sg12}
	\sum_{k=0}^{\infty} S_{k}(\mathbf{h}_{1m}{(v)}) \varepsilon^{k} =g_{1}^{[1,m]}(\varepsilon)\left( g_{2}^{[1,m]}(\varepsilon)\right)^{v}.
\end{equation}
In addition, let $ z=\frac{g_{2}^{[1,1]}(\varepsilon)\chi_{1}^{[1]}}{\chi_{1}-\chi_{1}^{*}}\varepsilon $, we can obtain $ \chi_{1}(\varepsilon) =\frac{\chi_{1}-z\chi_{1}^{*}}{1-z} $ by the definition of $ g_{2}^{[1,1]}(\varepsilon) $ in the equation \eqref{prop-notion},
and further generate
\begin{equation}\label{prop-gm1}
	\begin{aligned}
		&\frac{g_{1}^{[1,2]}(\varepsilon)}{g_{1}^{[1,1]}(\varepsilon)} = \frac{1}{1-\frac{\chi_{1}^{*}-\chi_{2}^{*}}{\chi_{1}-\chi_{2}^{*}}z} =\frac{1}{1-(C_{11}-C_{12}) g_{2}^{[1,1]}(\varepsilon) \varepsilon}, \\
		&g_{2}^{[1,2]}(\varepsilon)= \frac{(\chi_{1}-\chi_{1}^{*})z}{\left( 1-\frac{\chi_{1}^{*}-\chi_{2}^{*}}{\chi_{1}-\chi_{2}^{*}}z\right) \chi_{1}^{[1]}\varepsilon} = \frac{g_{2}^{[1,1]}(\varepsilon)}{1-(C_{11}-C_{12}) g_{2}^{[1,1]}(\varepsilon) \varepsilon},
	\end{aligned}
\end{equation}
with $ C_{mn} $ defined by the equation \eqref{tauij1}.
Then, we employ the above formulae \eqref{prop-sg12} and \eqref{prop-gm1} to calculate
\begin{equation}\label{prop-piki}
	\begin{aligned}
		f_{i}(\varepsilon)-1 =& \sum_{k=1}^{\infty} S_{k}(\mathbf{h}_{12}{(i-1)}) \varepsilon^{k} -\sum_{k=1}^{\infty} \sum_{v=0}^{k-1} S_{k-v}(\mathbf{h}_{11}{(i-1+v)}) p_{i,i+v}^{(v)}\varepsilon^{k} \\
		=& g_{1}^{[1,2]}(\varepsilon)\left( g_{2}^{[1,2]}(\varepsilon)\right)^{i-1}-1 - \sum_{v=0}^{\infty} p_{i,i+v}^{(v)}\varepsilon^{v} \left( g_{1}^{[1,1]}(\varepsilon)\left( g_{2}^{[1,1]}(\varepsilon)\right)^{i-1+v}-1\right)    \\
		=& g_{1}^{[1,1]}(\varepsilon)\left( g_{2}^{[1,1]}(\varepsilon)\right)^{i-1} \left( \left( \frac{1}{1-(C_{11}-C_{12}) g_{2}^{[1,1]}(\varepsilon) \varepsilon}\right)^{i} -f_{i}(g_{2}^{[1,1]}(\varepsilon)\varepsilon) \right) +f_{i}(\varepsilon)-1,
	\end{aligned}	
\end{equation}
which leads to
\begin{equation}\nonumber
	f_{i}(g_{2}^{[1,1]}(\varepsilon)\varepsilon)=\left( \frac{1}{1-(C_{11}-C_{12}) g_{2}^{[1,1]}(\varepsilon) \varepsilon}\right)^{i}. 
\end{equation}
As $ g_{2}^{[1,1]}(\varepsilon)=1+\cO(\varepsilon) $, we can set $ \hat{z}=g_{2}^{[1,1]}(\varepsilon)\varepsilon $, and obtain
\begin{equation}\label{prop-fi}
	f_{i}(\hat{z})=\left( \frac{1}{1-(C_{11}-C_{12}) \hat{z}}\right)^{i}. 
\end{equation}
By comparing the coefficients of the term of $ \hat{z}^{k} $ on both sides of the equation \eqref{prop-fi}, it is found that
\begin{equation}\label{prop-pijk}
	p_{i,i+k}^{(k)}= \binom{i+k-1}{i-1}(C_{11}-C_{12})^{k}, \quad k\geq 0,
\end{equation}
where $ \binom{k}{i} $ is the binomial coefficient.
Therefore, we can verify the equality of $ \hat{\mathbf{P}}_{21}\hat{\mathbf{P}}_{22}^{-1} $ in the equation \eqref{prop-hpli}.

On the other hand, the expression of $ \hat{\mathbf{P}}_{21}\hat{\mathbf{P}}_{22}^{-1} $ in the equation \eqref{prop-hpli} can be proved in the same manner, and the straightforward proof is omitted here.

This completes the proof of Lemma \ref{prop2}.

\subsection{Proof of Proposition \ref{prop3}}\label{subsec_proofp3}	
The proof of Proposition \ref{prop3} is provided in this subsection.

First, we will prove
\begin{equation}\label{prop-pn1n2}
	|\mathbf{P}_{N_{1},N_{2}}|=\frac{|C_{11}|^{N_{1}(N_{1}-1)}}{(\chi_{1}^{*}-\chi_{1})^{N_{1}}} \frac{|C_{22}|^{N_{2}(N_{2}-1)}}{(\chi_{2}^{*}-\chi_{2})^{N_{2}}} \gamma^{N_{1}N_{2}}, \quad \gamma=\left| \frac{\chi_{1}-\chi_{2}}{\chi_{1}^{*}-\chi_{2}}\right|^{2},
\end{equation}
for the arbitrary nonnegative integers $ N_{1} $ and $ N_{2} $.

For $ N_{1}>0 $ and $ N_{2}=0 $, we readily compute
\begin{equation}\label{propd-pn0}
	|\mathbf{P}_{N_{1},N_{2}}|=\left| \frac{1}{\chi_{1}^{*}-\chi_{1}}{\bar{\mathbf{P}}}_{11}^{\dagger}{\bar{\mathbf{P}}}_{11}\right| = 	\frac{|C_{11}|^{N_{1}(N_{1}-1)}}{(\chi_{1}^{*}-\chi_{1})^{N_{1}}},
\end{equation}
which verifies the equation \eqref{prop-pn1n2}. Similarly, the same argument applies to the case $ N_{1}=0 $ and $ N_{2}>0 $.

Next, we consider the case $ N_{1}>0 $ and $ N_{2}>0 $. Since the matrices $ \bar{\mathbf{P}}_{11} $ and $ \bar{\mathbf{P}}_{22} $ are invertible, we employ the Schur complement formula to yield
\begin{equation}\label{propd-pne0}
	\begin{aligned}
		&|\mathbf{P}_{N_{1},N_{2}}|= \frac{|C_{11}|^{N_{1}(N_{1}-1)}}{(\chi_{1}^{*}-\chi_{1})^{N_{1}}} \frac{|C_{22}|^{N_{2}(N_{2}-1)}}{(\chi_{2}^{*}-\chi_{2})^{N_{2}}}  \, \left| \mathbb{I}_{N_{2}} +(\gamma-1) \mathbf{R} \right|,	\\
		&\mathbf{R}=(\bar{\mathbf{P}}_{22}^{\dagger})^{-1} \bar{\mathbf{P}}_{21}^{\dagger} \bar{\mathbf{P}}_{12} \left( \bar{\mathbf{P}}_{11}^{\dagger}\bar{\mathbf{P}}_{11}\right)^{-1}  \bar{\mathbf{P}}_{12}^{\dagger} \bar{\mathbf{P}}_{21} (\bar{\mathbf{P}}_{22})^{-1}.
	\end{aligned}
\end{equation}

Moreover, from Proposition \ref{prop2}, we can obtain
\begin{equation}\label{prop-pin2}
	\begin{aligned}
		\left( \bar{\mathbf{P}}_{12}\bar{\mathbf{P}}_{11}^{-1}\right)_{i,j} &= \left( \mathbf{G}_{12}\hat{\mathbf{P}}_{12}\hat{\mathbf{P}}_{11}^{-1}\mathbf{G}_{11}^{-1}\right)_{i,j} = \binom{j-1}{i-1}(C_{11}-C_{12})^{j-i}\frac{C_{12}^{i-1}}{C_{11}^{j-1}},  \\
		\left( \bar{\mathbf{P}}_{21}\bar{\mathbf{P}}_{22}^{-1}\right)_{i,j} &= \left( \mathbf{G}_{21}\hat{\mathbf{P}}_{21}\hat{\mathbf{P}}_{22}^{-1}\mathbf{G}_{22}^{-1}\right)_{i,j} = \binom{j-1}{i-1}(C_{22}-C_{21})^{j-i}\frac{C_{21}^{i-1}}{C_{22}^{j-1}},
	\end{aligned} 
\end{equation}
where $ (P)_{i,j} $ denotes $ (i,j) $-entry of the matrix $ P $. Then, by the definition of $ C_{lm} $ in the equation \eqref{tauij1}, we can determine that 
\begin{equation}\label{prop-ppij}
	\left( ( \bar{\mathbf{P}}_{11}^{\dagger}) ^{-1}\bar{\mathbf{P}}_{12}^{\dagger}\bar{\mathbf{P}}_{21}\bar{\mathbf{P}}_{22}^{-1} \right)_{i,j}
	=\left(  \frac{\chi_{1}-\chi_{2}}{\chi_{1}^{*}-\chi_{2}}\right) ^{i-1} \left(  \frac{\chi_{1}^{*}-\chi_{2}^{*}}{\chi_{1}^{*}-\chi_{2}}\right) ^{j-1} \sum_{k=1}^{\min(i,j)} \binom{i-1}{k-1}\binom{j-1}{k-1} \left( 1-\frac{1}{\gamma}\right) ^{k-1},
\end{equation}
and thereby obtain
\begin{equation}\label{prop-ppma}
	\begin{aligned}
		&( \bar{\mathbf{P}}_{11}^{\dagger}) ^{-1}\bar{\mathbf{P}}_{12}^{\dagger}\bar{\mathbf{P}}_{21}\bar{\mathbf{P}}_{22}^{-1} = \bar{\mathbf{G}}_{1} \bar{\mathbf{R}} \bar{\mathbf{G}}_{2},\\
		&\bar{\mathbf{G}}_{1}=\diag\left( 1,\frac{\chi_{1}-\chi_{2}}{\chi_{1}^{*}-\chi_{2}} ,\left( \frac{\chi_{1}-\chi_{2}}{\chi_{1}^{*}-\chi_{2}}\right)^{2}, \ldots, \left( \frac{\chi_{1}-\chi_{2}}{\chi_{1}^{*}-\chi_{2}}\right)^{N_{1}-1}   \right),\\
		&\bar{\mathbf{G}}_{2}=  \diag\left( 1,\frac{\chi_{1}^{*}-\chi_{2}^{*}}{\chi_{1}^{*}-\chi_{2}} ,\left( \frac{\chi_{1}^{*}-\chi_{2}^{*}}{\chi_{1}^{*}-\chi_{2}}\right)^{2}, \ldots, \left( \frac{\chi_{1}^{*}-\chi_{2}^{*}}{\chi_{1}^{*}-\chi_{2}} \right)^{N_{2}-1}   \right),\\
		&\bar{\mathbf{R}}=\left( \sum_{k=1}^{\min(i,j)} \binom{i-1}{k-1}\binom{j-1}{k-1} \left( 1-\frac{1}{\gamma}\right) ^{k-1}\right)_{1\leq i\leq N_{1},1\leq j\leq N_{2}}.\\
	\end{aligned}
\end{equation} 
Then, from the definition of the matrix $ \mathbf{R} $ in the equation \eqref{propd-pne0}, we generate
\begin{equation}\label{prop-r}
	\bar{\mathbf{G}}_{2}\mathbf{R}\bar{\mathbf{G}}_{2}^{*}= (\hat{\mathbf{G}}_{2})^{2}\bar{\mathbf{R}}^{T}(\hat{\mathbf{G}}_{1})^{2}\bar{\mathbf{R}}(\hat{\mathbf{G}}_{2})^{2},
\end{equation}
with
\begin{equation}\label{prop-bg3}
	\hat{\mathbf{G}}_{m}=\diag\left( 1,\gamma^{1/2} ,\gamma, \ldots, \gamma^{(N_{m}-1)/2}   \right).
\end{equation}
Since $ \bar{\mathbf{G}}_{2}\bar{\mathbf{G}}_{2}^{*}=(\hat{\mathbf{G}}_{2})^{2} $, we can calculate
\begin{equation}\label{prop-matrix}
	\begin{aligned}
		\left| \mathbb{I}_{N_{2}} +(\gamma-1) \mathbf{R} \right| =& 
		\left| \mathbb{I}_{N_{2}} +(\gamma-1) \hat{\mathbf{G}}_{2}\bar{\mathbf{R}}^{T}(\hat{\mathbf{G}}_{1})^{2}\bar{\mathbf{R}}\hat{\mathbf{G}}_{2} \right|\\
		=& \left| \mathbb{I}_{N_{2}} +(\gamma-1)\hat{\mathbf{R}}^{T} \hat{\mathbf{R}}\right|\\
		=&\, \gamma^{N_{1}N_{2}},
	\end{aligned}
\end{equation}
based on Proposition \ref{prop1}, where $ \hat{\mathbf{R}} $ is defined by the equation \eqref{pl-deo-w}. Then, substituting the result of \eqref{prop-matrix} into \eqref{propd-pne0} verifies the equation \eqref{prop-pn1n2}.

Therefore, we can similarly calculate
\begin{equation}\label{prop-pn1n2b}
	\begin{aligned}
		|\mathbf{P}_{N_{1}-1,N_{2}}|&=\frac{|C_{11}|^{(N_{1}-1)(N_{1}-2)}}{(\chi_{1}^{*}-\chi_{1})^{N_{1}-1}} \frac{|C_{22}|^{N_{2}(N_{2}-1)}}{(\chi_{2}^{*}-\chi_{2})^{N_{2}}} \gamma^{(N_{1}-1)N_{2}},\\
		|\mathbf{P}_{N_{1}+1,N_{2}}|&=\frac{|C_{11}|^{(N_{1}+1)N_{1}}}{(\chi_{1}^{*}-\chi_{1})^{N_{1}+1}} \frac{|C_{22}|^{N_{2}(N_{2}-1)}}{(\chi_{2}^{*}-\chi_{2})^{N_{2}}} \gamma^{(N_{1}+1)N_{2}},\\
		|\mathbf{P}_{N_{1},N_{2}-1}|&=\frac{|C_{11}|^{N_{1}(N_{1}-1)}}{(\chi_{1}^{*}-\chi_{1})^{N_{1}}} \frac{|C_{22}|^{(N_{2}-1)(N_{2}-2)}}{(\chi_{2}^{*}-\chi_{2})^{N_{2}-1}} \gamma^{N_{1}(N_{2}-1)},\\
		|\mathbf{P}_{N_{1},N_{2}+1}|&=\frac{|C_{11}|^{N_{1}(N_{1}-1)}}{(\chi_{1}^{*}-\chi_{1})^{N_{1}}} \frac{|C_{22}|^{(N_{2}+1)N_{2}}}{(\chi_{2}^{*}-\chi_{2})^{N_{2}+1}} \gamma^{N_{1}(N_{2}+1)},\\
	\end{aligned}
\end{equation}
with $ N_{1}>0 $ and $ N_{2}>0 $, and omit the calculation process here. Then, substituting \eqref{prop-pn1n2} and \eqref{prop-pn1n2b} into \eqref{dete-p} confirms the identity in the equation \eqref{dete-p}.

The proof of Proposition \ref{prop3} is completed.

\subsection{Proof of Theorem \ref{theo-nmrp}}\label{subsec_proofnmrp}	
This subsection presents the proof of Theorem \ref{theo-nmrp}.

First, we will rewrite the determinant $ \tau^{(k)} $ $ (k=0,1,2) $ in the $ [N_{1}, N_{2}] $-order RW solution formula \eqref{horw}. Based on the definitions of the matrices $ \mathbf{M}^{(k,l,m)} $ in the equations \eqref{msmat1} and \eqref{tauij1}, the matrices $ \mathbf{M}^{(k,l,m)} $ are also expressed as
	\begin{equation}\label{Mslm1}
		\mathbf{M}^{(k,l,m)}= \frac{1}{\chi_{l}^{*}-\chi_{m}} \mathbf{M}^{(k,-)}_{lm}\mathbf{M}^{(k,+)}_{ml}, \quad k=0,1,2, \quad l,m=1,2,
	\end{equation}
	where 
	\begin{equation}\label{Mspm1}
	\begin{aligned}
		&\mathbf{M}^{(k,+)}_{ml}=\left[ C_{ml}^{i-1}S_{2j-i}( \mathbf{u}^{(k,m,l,+)}(i-1))  \right]_{1\leq i\leq 2N_{ml}, 1\leq j\leq N_{m}},\\
		&\mathbf{M}^{(k,-)}_{lm}=\left[ (C_{lm}^{*})^{j-1}S_{2i-j}( \mathbf{u}^{(k,l,m,-)}(j-1))  \right]_{1\leq i\leq N_{l}, 1\leq j\leq 2N_{lm}},\\
	\end{aligned}
	\end{equation}
	the symbol $ N_{lm}$ is defined by the equation \eqref{utm}, $C_{ml}  $ is given in the equation \eqref{tauij1}, and the vectors 
	\[\mathbf{u}^{(k,l,m,\pm)}(v)=\left( {u}_{1}^{(k,l,m,\pm)}(v), {u}_{2}^{(k,l,m,\pm)}(v),\ldots\right)\]
	are defined by
	\begin{equation}\label{uslm}
		\begin{array}{ll}
			\mathbf{u}^{(k,m,l,+)}(v)= \mathbf{x}^{(k,+)}_{m} + \mathbf{h}_{ml}{(v)}, & \mathbf{x}^{(k,+)}_{m}=\mathbf{x}^{+}_{m}  +\mathbf{h}_{3}^{[m,k]}, \\
			\mathbf{u}^{(k,l,m,-)}(v)=\mathbf{x}^{(k,-)}_{l} + \mathbf{h}_{lm}^{*}{(v)}, & \mathbf{x}^{(k,-)}_{l}=\mathbf{x}^{-}_{l} - (\mathbf{h}_{3}^{[l,k]})^{*},
		\end{array}
	\end{equation}
	with $ \mathbf{h}_{3}^{[l,0]} $ being a zero vector, and $\mathbf{x}^{\pm}_{l}$, $ \mathbf{h}_{j}^{[l,m]} $ $ (1\leq j\leq 3) $, and $ \mathbf{h}_{ml}{(v)} $ given in the equations \eqref{xpm1}, \eqref{xpm2}, and \eqref{utm}, respectively.
	
	Since
	\begin{equation}\label{Schp1}
	\begin{aligned}
		\sum_{i=0}^{\infty} S_{i}(\mathbf{u}^{(k,m,l,+)}(v) ) \varepsilon^{i} &= \exp\left( \sum_{j=1}^{\infty} {u}_{j}^{(k,l,m,+)}(v)\varepsilon^{j} \right) \\
		&= \exp\left(\sum_{j=1}^{\infty} ({x}^{+}_{m,j} +h_{3,j}^{[m,k]}) \varepsilon^{j} \right) \exp\left( \sum_{j=1}^{\infty} ({h}_{1,j}^{[m,l]} +v{h}_{2,j}^{[m,l]}) \varepsilon^{j} \right) \\
		&= \left( \sum_{i=0}^{\infty} S_{i}( \mathbf{x}^{(k,+)}_{m}  ) \varepsilon^{i}\right)
		 \, \left(  \sum_{j=0}^{\infty} S_{j}(\mathbf{h}_{ml}{(v)}) \varepsilon^{j}\right)\\
		&= \sum_{i=0}^{\infty} \left(\sum_{j=0}^{i} S_{i-j}( \mathbf{x}^{(k,+)}_{m} ) S_{j}(\mathbf{h}_{ml}{(v)}) \right) \varepsilon^{i},
	\end{aligned}
	\end{equation}
	we have
	\begin{equation}\label{Schexp1}
		S_{i}(\mathbf{u}^{(k,m,l,+)}(v) ) =\sum_{j=0}^{i} S_{i-j}(\mathbf{x}^{(k,+)}_{m} ) S_{j}(\mathbf{h}_{ml}{(v)}).
	\end{equation}
	Similarly, we can generate
	\begin{equation}\label{Schexp2}
		S_{i}(\mathbf{u}^{(k,l,m,-)}(v) ) =\sum_{j=0}^{i} S_{i-j}( \mathbf{x}^{(k,-)}_{l} ) S_{j}(\mathbf{h}_{lm}^{*}{(v)}).
	\end{equation}
	We apply the formulae \eqref{Schexp1} and \eqref{Schexp2} to expand all elements of the matrices $ \mathbf{M}^{(k,\pm)}_{lm} $, and further obtain the expressions:
	\begin{equation}\label{Mslm2}
		\mathbf{M}^{(k,+)}_{ml}= \mathbf{P}_{ml} \,{\mathbf{M}}_{m}^{(k,+)}, 	\quad \mathbf{M}^{(k,-)}_{lm}= {\mathbf{M}}_{l}^{(k,-)}\, \mathbf{P}_{lm}^{\dagger} ,\\
	\end{equation}
	with
	\begin{equation}\label{Mslm2a}
	\begin{aligned}
		{\mathbf{M}}_{m}^{(k,+)} =& \left[ S_{2j-i}(\mathbf{x}^{(k,+)}_{m}) \right]_{1\leq i\leq 2N_{m}, 1\leq j\leq N_{m}}, \quad
		{\mathbf{M}}_{l}^{(k,-)} = \left[ S_{2i-j}( \mathbf{x}^{(k,-)}_{l} ) \right]_{1\leq i\leq N_{l}, 1\leq j\leq 2N_{l}},\\
		\mathbf{P}_{ml}\quad=& \left[C_{ml}^{i-1} S_{j-i}(\mathbf{h}_{ml}{(v)}) \right]_{1\leq i\leq 2N_{lm}, 1\leq j\leq 2N_{m}}. 	
	\end{aligned}
	\end{equation}

	Thus, the determinant $ \tau^{(k)} $ can be rewritten as
	\begin{equation}\label{taus1}
		\tau^{(k)}=	\det\left( 
		\begin{bmatrix}
			{\mathbf{M}}_{1}^{(k,-)} & \mathbf{0}\\
			\mathbf{0} & {\mathbf{M}}_{2}^{(k,-)}
		\end{bmatrix}
		\mathbf{P}
		\begin{bmatrix}
			{\mathbf{M}}_{1}^{(k,+)} & \mathbf{0}\\
			\mathbf{0} & {\mathbf{M}}_{2}^{(k,+)}
		\end{bmatrix}
		\right), \quad k=0,1,2,
	\end{equation}
	with
	\begin{equation}\label{mat_p}
		\mathbf{P}=\begin{bmatrix}
			\frac{1}{\chi_{1}^{*}-\chi_{1}}\mathbf{P}_{11}^{\dagger}\mathbf{P}_{11} & \frac{1}{\chi_{1}^{*}-\chi_{2}}\mathbf{P}_{12}^{\dagger}\mathbf{P}_{21}\\
			\frac{1}{\chi_{2}^{*}-\chi_{1}}\mathbf{P}_{21}^{\dagger}\mathbf{P}_{12} & \frac{1}{\chi_{2}^{*}-\chi_{2}}\mathbf{P}_{22}^{\dagger}\mathbf{P}_{22}
		\end{bmatrix}.
	\end{equation}
	
	Then, by utilizing the Cauchy--Binet formula, the determinant $ \tau^{(k)} $ in the equation \eqref{taus1} can be expanded as the following form
	\begin{equation}\label{taus_ex1}
		\tau^{(k)}= \sum_{\substack{ 1\leq v_{l,1}< v_{l,2}<\ldots< v_{l,N_{l}}\leq 2N_{l},  \\ 1\leq \bar{v}_{l,1}< \bar{v}_{l,2}<\ldots< \bar{v}_{l,N_{l}}\leq 2N_{l}, }}
		\left| \mathbf{P}[\mathbf{r},\mathbf{s}] \right| 
		\left| 	{\mathbf{M}}_{1}^{(k,-)} [:,\mathbf{v}_{1}] \right| 
		\left| 	{\mathbf{M}}_{2}^{(k,-)} [:,\mathbf{v}_{2}] \right| 
		\left| 	{\mathbf{M}}_{1}^{(k,+)} [\bar{\mathbf{v}}_{1},:]\right| 
		\left| 	{\mathbf{M}}_{2}^{(k,+)} [\bar{\mathbf{v}}_{2},:] \right| ,
	\end{equation}
	where $ {\mathbf{M}}_{l}^{(k,-)}[:,\mathbf{v}_{l}] $ denotes the $ N_{l}\times N_{l} $ submatrices of $  {\mathbf{M}}_{l}^{(k,-)} $ formed by selecting the columns indexed by the vectors $ \mathbf{v}_{l}=(v_{l,1}, v_{l,2},\ldots, v_{l,N_{l}}) $, $ {\mathbf{M}}_{l}^{(k,+)}[\bar{\mathbf{v}}_{l},:] $ denotes the $ N_{l}\times N_{l} $ submatrices of $  {\mathbf{M}}_{l}^{(k,+)} $ formed by the rows indexed by the vectors $ \bar{\mathbf{v}}_{l}=(\bar{v}_{l,1}, \bar{v}_{l,2},\ldots, \bar{v}_{l,N_{l}}) $, and $ \mathbf{P}[\mathbf{r},\mathbf{s}] $ denotes the $ N\times N $ submatrix of $ \mathbf{P} $ by selecting the combined row index vector $ \mathbf{r}=({v}_{1,1},{v}_{2,1},\ldots,v_{1,N_{1}},{v}_{2,1}+2N_{1},{v}_{2,2}+2N_{1},\ldots,{v}_{2,N_{2}}+2N_{1}) $ and the combined column index vector $ \mathbf{s}=(\bar{v}_{1,1},\bar{v}_{2,1},\ldots,\bar{v}_{1,N_{1}},\bar{v}_{2,1}+2N_{1},\bar{v}_{2,2}+2N_{1},\ldots,\bar{v}_{2,N_{2}}+2N_{1}) $.
	
	When the internal large parameters $\{d_{l,2j+1}\}_{1\leq l \leq 2, 0 \leq j\leq N_{l}-1}$ are defined by the equation \eqref{dj1}, we can generate the asymptotic expression
	\begin{equation}\label{asysk1}
		\begin{aligned}
			&S_{j}( \mathbf{x}^{(k,+)}_{l} ) \sim A^{j}\varphi_{j}(z_{l},\bm{\kappa}_{l}), \quad 
			S_{j}( \mathbf{x}^{(k,-)}_{l} ) \sim A^{j}\varphi_{j}^{*}(z_{l},\bm{\kappa}_{l}), \quad \text{as} \quad A\gg 1,\\
			&z_{l}=x_{l,1}^{+}A^{-1}=\ii \chi_{l}^{[1]}\left( x+\chi_{l}t \right)A^{-1} + \kappa_{l,1},
		\end{aligned}
	\end{equation}
	by a calculation similar to that in the equation \eqref{Schp1}.
	It is found that the order of $A$ in the polynomials $S_{j}( \mathbf{x}^{(k,+)}_{l} )$ and $ S_{j}( \mathbf{x}^{(k,-)}_{l} ) $ decreases as the subscript $j$ decreases.
	 
	 Then, we can calculate the highest-order term of $A$ of $ \tau^{(k)} $ in the equation \eqref{taus_ex1} by several index vectors to analyze the asymptotics of $\mathbf{q}^{[N_{1},N_{2}]}(x,t)$ with $A\gg1$, where the index vectors are 
	 \begin{equation}\label{indv1}
	 \begin{array}{ll}
	 	\mathbf{v}_{l}^{(1)}=(1,2, \ldots, N_{l}-1, N_{l}), & 
	 	\mathbf{v}_{l}^{(2)}=(1,2, \ldots, N_{l}-1, N_{l}+1), \quad l=1,2. \\
	 \end{array}
	 \end{equation} 
	 
	When choosing a set of index vectors $ \mathbf{v}_{l}=\bar{\mathbf{v}}_{l}=\mathbf{v}_{l}^{(1)} $ $ (1\leq l\leq 2) $, we can calculate the highest-order term of $A$ in $ \tau^{(k)} $ and provide the following expression:
 	\begin{equation}\label{taus_hot1}
 		\tau^{(k)}=c_{N_{1}}^{-2}c_{N_{2}}^{-2} \left| \mathbf{P}[\mathbf{r}_{11},\mathbf{s}_{11}] \right| \left| \Theta_{N_{1}}(z_{1},\bm{\kappa}_{1})\right|^{2} \left| \Theta_{N_{2}}(z_{2},\bm{\kappa}_{2})\right|^{2} A^{\Gamma}+ \cO(A^{\Gamma-1}),
 	\end{equation}
	where $ \Gamma=\sum_{l=1}^{2}{N_{l}(N_{l}+1)} $, $ \mathbf{P}[\mathbf{r}_{11},\mathbf{s}_{11}] $ is a $ N\times N $ constant matrix formed by the row index vector $ \mathbf{r}_{11} =\mathbf{r} |_{\mathbf{v}_{1}=\mathbf{v}_{1}^{(1)},\mathbf{v}_{2}=\mathbf{v}_{2}^{(1)}} $ and the column index vector $ \mathbf{s}_{11}= \mathbf{s}|_{\bar{\mathbf{v}}_{1}=\mathbf{v}_{1}^{(1)}, \bar{\mathbf{v}}_{2} =\mathbf{v}_{2}^{(1)}} $ of $ \mathbf{P} $, and the constant $ c_{N_{l}} $ is given in the equation \eqref{amp}.
	
	Substituting the above expression \eqref{taus_hot1} into the RW solution formula \eqref{horw}, then we find that as $ A\rightarrow \infty $,
	\begin{equation}\label{qasy1}
		\mathbf{q}^{[N_{1},N_{2}]}(x,t)\rightarrow\mathbf{q}^{[0]}(x,t),
	\end{equation}
	except for some locations near $ (x_{l,0}, t_{l,0}) $ of the $ (x,t) $-plane. Here, $ (x_{l,0}, t_{l,0}) $ is defined by
	\begin{equation}\label{x1}
		\ii \chi^{[1]}_{l} (x_{l,0} +\chi_{l}t_{l,0})A^{-1}+\kappa_{l,0}=z_{l,0},  \quad l=1,2,
	\end{equation} 
	with $ z_{l,0} $ being each of simple root of the Adler--Moser polynomial $ \Theta_{N_{l}}(z_{l},\bm{\kappa}_{l}) $. Now, we complete the proof of Item $ (3) $ of Theorem \ref{theo-nmrp}.
	
	Next, we consider the asymptotics of the RW solution $ \mathbf{q}^{[N_{1}, N_{2}]}(x,t) $ near the location $ (x_{l,0}, t_{l,0}) $ of the $ (x,t) $-plane.
	Here, we provide the asymptotic analysis of $ \mathbf{q}^{[N_{1}, N_{2}]}(x,t) $ near the location $ (x_{1,0}, t_{1,0}) $. Then, the asymptotic behavior near the location $ (x_{2,0}, t_{2,0}) $ can be established in a similar manner and is therefore omitted.
	
	When $ N_{1}>0 $, from the assumptions of Theorem \ref{theo-nmrp}, we find that all roots of $ \Theta_{N_{1}}(z_{1},\bm{\kappa}_{1}) $ are simple, and there exists a constant $ c>0 $ such that 
	$ \|(x_{1,0},t_{1,0})-(x_{2,0},t_{2,0})\|_{2}\geq c $.
	Thus, it suffices to compute the leading two terms of $ A $ of $ \tau^{(k)} $ in the equation \eqref{taus_ex1} by the indexes $ \mathbf{v}_{2}=\bar{\mathbf{v}}_{2}=\mathbf{v}_{2}^{(1)} $ and the four possible choices for the other two index vectors: 
	\begin{equation}\label{indv2}
		\begin{aligned}
			&(i):\quad \mathbf{v}_{1}=\bar{\mathbf{v}}_{1}=\mathbf{v}_{1}^{(1)}; \qquad 
			&(ii): \quad \mathbf{v}_{1}=\mathbf{v}_{1}^{(1)},\bar{\mathbf{v}}_{1}=\mathbf{v}_{1}^{(2)}; \\
			&(iii): \quad \mathbf{v}_{1}=\mathbf{v}_{1}^{(2)},\bar{\mathbf{v}}_{1}=\mathbf{v}_{1}^{(1)}; \qquad
			&(iv): \quad \mathbf{v}_{1}=\mathbf{v}_{1}^{(2)},\bar{\mathbf{v}}_{1}=\mathbf{v}_{1}^{(2)};
		\end{aligned}
	\end{equation}
	with $ \mathbf{v}_{1}^{(l)} $ $ (1\leq l\leq 2) $ defined in the equation \eqref{indv1}.
	
	Now, we take a coordinate transformation
	\begin{equation}\label{hatxt1}
		\hat{x}=x-\hat{x}_{0}A, \quad \hat{t}=t-\hat{t}_{0}A,
	\end{equation}
	where $ (\hat{x}_{0}, \hat{t}_{0})=(x_{1,0}, t_{1,0})A^{-1} $, and $ (x_{1,0},t_{1,0}) $ are defined by the equation \eqref{x1} with a simple root of $ \Theta_{N_{1}}(z_{1},\bm{\kappa}_{1}) $ at the $ (x,t) $-plane. Then, we have
	\begin{equation}\label{uxtv}
		\begin{aligned}
			&\mathbf{x}^{(k,\pm)}_{l}(x,t) = {\mathbf{x}}_{l}^{(k,\pm)}(\hat{x},\hat{t}) +\mathbf{x}_{l}^{\pm}(\hat{x}_{0}A,\hat{t}_{0}A),  \quad l=1,2, \quad 0\leq s \leq 2,
		\end{aligned}
	\end{equation}
	where the vectors $ \mathbf{x}_{l}^{\pm}(x,t)=({{x}}_{l,1}^{\pm}({x},{t}), {{x}}_{l,2}^{\pm}({x},{t}), \ldots) $ and $ \mathbf{x}_{l}^{(k,\pm)}(x,t)=({{x}}_{l,1}^{(k,\pm)}({x},{t}), {{x}}_{l,2}^{(k,\pm)}({x},{t}), \ldots)  $ are given in the equations \eqref{xpm1} and \eqref{uslm}, respectively. When the large parameters $ d_{l,2j}=0 $ and $ d_{l,2j+1} $ are defined by the equation \eqref{dj1} with $ 1\leq l\leq 2$ and $  j\geq0 $, we yield 
	\begin{equation}\label{asyscp}
		\begin{aligned}
			 &\sum_{i=0}^{\infty} S_i (\mathbf{x}^{(k,+)}_{l}(x,t))(A^{-1} \varepsilon)^i \\
			=&\exp \left( \sum_{j=1}^{\infty} \left( {x}_{l,j}^{(k,+)}(\hat{x},\hat{t}) +{x}_{l,j}^{+}(\hat{x}_{0}A,\hat{t}_{0}A)\right) (A^{-1} \varepsilon)^j \right)\\
			=& \exp\left( (\ii\chi_{l}^{[1]}(\hat{x}_{0} + \chi_{l}\hat{t}_{0}) + \kappa_{l,1})\varepsilon + \sum_{j=1}^{\infty}\kappa_{l,2j+1}\varepsilon^{2j+1} \right)  \exp\left( \left( \hat{x}_{l,1}^{(k,+)}(\hat{x},\hat{t})\varepsilon + {x}_{l,2}^{+}(\hat{x}_{0},\hat{t}_{0})\varepsilon^{2} \right) A^{-1} +\cO(A^{-2}) \right) \\
			=& \left(\sum_{j=0}^{\infty} \varphi_{j}(\hat{z}_{l,0}, \bm{\kappa}_{l}) \varepsilon^{j}\right) \left[1+ \left( \hat{x}_{l,1}^{(k,+)}(\hat{x},\hat{t})\varepsilon + {x}_{l,2}^{+}(\hat{x}_{0},\hat{t}_{0})\varepsilon^{2}\right) A^{-1} +\cO(A^{-2}) \right]
		\end{aligned}
	\end{equation}
	with
	\begin{equation}\label{zlhx}
		 \hat{z}_{l,0}=\ii\chi_{l}^{[1]}(\hat{x}_{0} + \chi_{l}\hat{t}_{0}) + \kappa_{l,1}, \quad \hat{x}_{l,1}^{(k,+)}(\hat{x},\hat{t})= \ii\chi_{l}^{[1]}\left( \hat{x}+\chi_{l}\hat{t}\right)  +h_{3,1}^{[l,k]},
	\end{equation}
	and then obtain
	\begin{equation}\label{asyscp1}
	\begin{aligned}
		S_{j}(\mathbf{x}_{l}^{(k,+)}(x,t)) =& A^j \left[ \varphi_{j}(\hat{z}_{l,0},\bm{\kappa}_{l}) +A^{-1} \left(  \hat{x}_{l,1}^{(k,+)}(\hat{x},\hat{t}) \varphi_{j-1}(\hat{z}_{l,0},\bm{\kappa}_{l}) + {x}_{l,2}^{+}(\hat{x}_{0},\hat{t}_{0}) \varphi_{j-2}(\hat{z}_{l,0},\bm{\kappa}_{l})\right) +\cO(A^{-2})\right].
	\end{aligned}
	\end{equation}
	Similarly, we also generate
	\begin{equation}\label{asyscp2}
	\begin{aligned}
		S_{j}(\mathbf{x}_{l}^{(k,-)}(x,t)) =& A^j \left[ \varphi_{j}^{*}(\hat{z}_{l,0},\bm{\kappa}_{l}) +A^{-1} \left(  \hat{x}_{l,1}^{(k,-)}(\hat{x},\hat{t}) \varphi_{j-1}^{*}(\hat{z}_{l,0},\bm{\kappa}_{l}) + {x}_{l,2}^{-}(\hat{x}_{0},\hat{t}_{0}) \varphi_{j-2}^{*}(\hat{z}_{l,0},\bm{\kappa}_{l})\right) +\cO(A^{-2})\right]
	\end{aligned}
	\end{equation}
	with
	\begin{equation}\label{xli1}
		\hat{x}_{l,1}^{(k,-)}(\hat{x},\hat{t})= -\ii(\chi_{l}^{[1]})^{*}\left( \hat{x}+\chi_{l}^{*}\hat{t}\right)  -(h_{3,1}^{[l,k]})^{*}.
	\end{equation}

	Based on the assumptions of Theorem \ref{theo-nmrp} and $ \hat{z}_{1,0}=z_{1,0} $, it is found that
	\begin{equation}\label{the120}
		 \Theta_{N_{1}}(\hat{z}_{1,0},\bm{\kappa}_{1})=0 , \quad \Theta_{N_{1}}'(\hat{z}_{1,0},\bm{\kappa}_{1})\ne0, \quad \Theta_{N_{2}}(\hat{z}_{2,0},\bm{\kappa}_{2})\ne0.
	\end{equation}
Thus, for the index vectors $\mathbf{v}_2=\bar{\mathbf{v}}_2=\mathbf{v}_2^{(1)}$ and those in the equation \eqref{indv2}, we can yield
	\begin{equation}\label{asyms1}
	\begin{aligned}
		&\left| M_{1}^{(k,-)}[:,\mathbf{v}_1^{(1)}]\right|=
		c_{N_{1}}^{-1}\left[ \hat{x}_{1,1}^{(k,-)}(\hat{x},\hat{t})  \left(\Theta_{N_1}^{*}(\hat{z}_{1,0},\bm{\kappa}_1)\right)' +{x}_{1,2}^{-}(\hat{x}_{0},\hat{t}_{0}) \mathcal{L}_{1}\Theta_{N_{1}}^{*}(\hat{z}_{1,0},\bm{\kappa}_{1}) \right]  A^{\frac{N_1(N_1 + 1)}{2}-1} +\mathcal{O}(A^{\frac{N_1(N_1 + 1)}{2}-2}),\\
		&\left| M_{1}^{(k,+)}[\bar{\mathbf{v}}_1^{(1)},:]\right|=
		c_{N_{1}}^{-1}\left[ \hat{x}_{1,1}^{(k,+)}(\hat{x},\hat{t}) \Theta_{N_1}'(\hat{z}_{1,0},\bm{\kappa}_1) +{x}_{1,2}^{+}(\hat{x}_{0},\hat{t}_{0}) \mathcal{L}_{1}\Theta_{N_{1}}(\hat{z}_{1,0},\bm{\kappa}_{1}) \right]  A^{\frac{N_1(N_1 + 1)}{2}-1} +\mathcal{O}(A^{\frac{N_1(N_1 + 1)}{2}-2}),\\
		&\left| M_{1}^{(k,-)}[:,\mathbf{v}_1^{(2)}]\right|=
		c_{N_{1}}^{-1} \Theta_{N_1}^*(\hat{z}_{1,0},\bm{\kappa}_1)\, A^{\frac{N_1(N_1 + 1)}{2}-1} +\mathcal{O}(A^{\frac{N_1(N_1 + 1)}{2}-2}),\\
		&\left| M_{1}^{(k,+)}[\bar{\mathbf{v}}_1^{(2)},:]\right|=
		c_{N_{1}}^{-1}  \Theta_{N_1}(\hat{z}_{1,0},\bm{\kappa}_1)\, A^{\frac{N_1(N_1 + 1)}{2}-1} +\mathcal{O}(A^{\frac{N_1(N_1 + 1)}{2}-2}),\\
		&\left| M_{2}^{(k,-)}[:,\mathbf{v}_2^{(1)}]\right|=
		c_{N_{2}}^{-1} \Theta_{N_2}^*(\hat{z}_{2,0},\bm{\kappa}_2)\, A^{\frac{N_2(N_2 + 1)}{2}} +\mathcal{O}(A^{\frac{N_2(N_2 + 1)}{2}-1}),\\
		&\left| M_{2}^{(k,+)}[\bar{\mathbf{v}}_2^{(1)},:]\right|=
		c_{N_{2}}^{-1}  \Theta_{N_2}(\hat{z}_{2,0},\bm{\kappa}_2)\, A^{\frac{N_2(N_2 + 1)}{2}} +\mathcal{O}(A^{\frac{N_2(N_2 + 1)}{2}-1}),\\
	\end{aligned}
	\end{equation}
	by the formulae \eqref{asyscp1} and \eqref{asyscp2} with the operator $ \mathcal{L}_{1} $ defined by Definition \ref{defi_1}. Then, according to Proposition \ref{prop-app1}, we have
	\[\mathcal{L}_{1} \Theta_{N_{1}}(\hat{z}_{1,0},\bm{\kappa}_{1})=\mathcal{L}_{2} \Theta_{N_{1}}(\hat{z}_{1,0},\bm{\kappa}_{1})=0. \]

Therefore, when $ A\gg 1 $ and $\sqrt{\hat{x}^{2}+\hat{t}^{2}} = \cO(1)$ (i.e., $\sqrt{(x-{x}_{1,0})^{2}+(t-{t}_{1,0})^{2}} = \cO(1)$), we can simplify the expression of $ \tau^{(k)} $ in the equation \eqref{taus_ex1} to the following asymptotic form
	\begin{equation}\label{asyexpre1}
	\begin{aligned}
		\tau^{(k)}& =\left[\left| \mathbf{P}[\mathbf{r}_{21},\mathbf{s}_{21}]\right| +\left| \mathbf{P}[\mathbf{r}_{11},\mathbf{s}_{11}]\right|  \hat{x}_{1,1}^{(k,+)}(\hat{x},\hat{t})  \hat{x}_{1,1}^{(k,-)}(\hat{x},\hat{t}) +\left| \mathbf{P}[\mathbf{r}_{11},\mathbf{s}_{21}]\right|  \hat{x}_{1,1}^{(k,-)}(\hat{x},\hat{t})   + \left| \mathbf{P}[\mathbf{r}_{21},\mathbf{s}_{11}]\right|  \hat{x}_{1,1}^{(k,+)}(\hat{x},\hat{t})   	\right]\\
%		&\left. \\
		&\times 
		c_{N_{1}}^{-2}c_{N_{2}}^{-2}\left|\Theta_{N_1}'(\hat{z}_{1,0},\bm{\kappa}_1)\right|^2 \left|\Theta_{N_2}(\hat{z}_{2,0},\bm{\kappa}_2)\right|^2 A^{\Gamma-2}
		+\cO(A^{\Gamma-3}) \\
		&=  \left| \mathbf{P}[\mathbf{r}_{11},\mathbf{s}_{11}]\right|  
		\left[ \left( \hat{x}_{1,1}^{(k,+)}(\hat{x},\hat{t})  +\Delta_{1}(\hat{z}_{1,0}) \right) \left( \hat{x}_{1,1}^{(k,-)}(\hat{x},\hat{t})  +\bar{\Delta}_{1}(\hat{z}_{1,0}) \right) 
		+\dfrac{ \left| \mathbf{P}[\mathbf{r}_{21},\mathbf{s}_{21}]\right|}{ \left| \mathbf{P}[\mathbf{r}_{11},\mathbf{s}_{11}]\right|} 
		-\dfrac{ \left| \mathbf{P}[\mathbf{r}_{21},\mathbf{s}_{11}]\right|}{ \left| \mathbf{P}[\mathbf{r}_{11},\mathbf{s}_{11}]\right|} 
		\dfrac{ \left| \mathbf{P}[\mathbf{r}_{11},\mathbf{s}_{21}]\right|}{ \left| \mathbf{P}[\mathbf{r}_{11},\mathbf{s}_{11}]\right|} \right] \\ 
		&  \times c_{N_{1}}^{-2}c_{N_{2}}^{-2}    \left|\Theta_{N_1}'(\hat{z}_{1,0},\bm{\kappa}_1)\right|^2 \left|\Theta_{N_2}(\hat{z}_{2,0},\bm{\kappa}_2)\right|^2 A^{\Gamma-2}
		+\cO(A^{\Gamma-3}), \quad k=0,1,2,
	\end{aligned}
	\end{equation}
	based on the four sets of indexes \eqref{indv2} and the formulae \eqref{asyms1}, where the complex constants $ \Delta_{1}(\hat{z}_{1,0})=\cO(1) $ and $ \bar{\Delta}_{1}(\hat{z}_{1,0})=\cO(1) $ are defined by
	\begin{equation}\label{delta1}
		\begin{aligned}
			\Delta_{1}= \dfrac{ \left| \mathbf{P}[\mathbf{r}_{11},\mathbf{s}_{21}]\right|}{ \left| \mathbf{P}[\mathbf{r}_{11},\mathbf{s}_{11}]\right|} , \quad
			\bar{\Delta}_{1}= \dfrac{ \left| \mathbf{P}[\mathbf{r}_{21},\mathbf{s}_{11}]\right|}{ \left| \mathbf{P}[\mathbf{r}_{11},\mathbf{s}_{11}]\right|} , 
		\end{aligned}
	\end{equation}
	the row indexes $ \mathbf{r}_{m1}= \mathbf{r}|_{\mathbf{v}_{1} =\mathbf{v}_{1}^{(m)}, \mathbf{v}_{2}=\mathbf{v}_{2}^{(1)}}$, and the column indexes $ \mathbf{s}_{m1}= \mathbf{s}|_{\bar{\mathbf{v}}_{1} =\mathbf{v}_{1}^{(m)}, \bar{\mathbf{v}}_{2}=\mathbf{v}_{2}^{(1)}}$ with $ m=1,2 $.
	Based on the definition of the elements in the matrix $ \mathbf{P} $ \eqref{mat_p}, we verify the equality of
	\[ \left( \dfrac{ \left| \mathbf{P}[\mathbf{r}_{11},\mathbf{s}_{21}]\right|}{ \left| \mathbf{P}[\mathbf{r}_{11},\mathbf{s}_{11}]\right|} \right)^{*} = \dfrac{ \left| \mathbf{P}[\mathbf{r}_{11},\mathbf{s}_{21}]\right|}{ \left| \mathbf{P}[\mathbf{r}_{11},\mathbf{s}_{11}]\right|},
	\]
	and thereby concludes that $ \Delta_{1}= \bar{\Delta}_{1}^{*}$. In addition, by the Desnanot--Jacobi identity \cite{behrend2013} and the Proposition \ref{prop3}, we can obtain
	\begin{equation}\label{cprs}
		\begin{aligned}
			\dfrac{ \left| \mathbf{P}[\mathbf{r}_{21},\mathbf{s}_{21}]\right|}{ \left| \mathbf{P}[\mathbf{r}_{11},\mathbf{s}_{11}]\right|} 
			-\dfrac{ \left| \mathbf{P}[\mathbf{r}_{21},\mathbf{s}_{11}]\right|}{ \left| \mathbf{P}[\mathbf{r}_{11},\mathbf{s}_{11}]\right|} 
			\dfrac{ \left| \mathbf{P}[\mathbf{r}_{11},\mathbf{s}_{21}]\right|}{ \left| \mathbf{P}[\mathbf{r}_{11},\mathbf{s}_{11}]\right|} &=\frac{|\mathbf{P}_{N_{1}+1,N_{2}}|
				|\mathbf{P}_{N_{1}-1,N_{2}}|}{|\mathbf{P}_{N_{1},N_{2}}|^{2}}\\
			&=|C_{11}|^{2},
		\end{aligned}
	\end{equation} 
	where the matrix $ \mathbf{P}_{N_{1},N_{2}} $ denotes the $ (N_{1}+N_{2})\times(N_{1}+N_{2}) $ principal submatrix of $ \mathbf{P} $ by selecting the rows and columns indexed by $ (1,2,\ldots,N_{1}, 2N_{1}+1,2N_{1}+2, \ldots, 2N_{1}+N_{2}) $.
	
	Since $ (\hat{x},\hat{t})=(x-x_{1,0}, t-t_{1,0}) $, we can absorb the constant term $ \Delta_{1}(\hat{z}_{0,1})=\Delta_{1}({z}_{0,1}) $ into $ (x_{1,0}, t_{1,0}) $,
	and then derive
	\begin{equation}\label{asyexpre2}
		\begin{aligned}
			\tau^{(k)} = &  c_{N_{1}}^{-2}c_{N_{2}}^{-2} \left| \mathbf{P}[\mathbf{r}_{11},\mathbf{s}_{11}]\right| 
			\left[ \hat{x}_{1,1}^{(k,+)}(x-x_{1,0},t-t_{1,0})   \hat{x}_{1,1}^{(k,-)}(x-x_{1,0},t-t_{1,0})   +|C_{11}|^{2} \right]\\ 
			&  \times  \left|\Theta_{N_1}'(\hat{z}_{1,0},\bm{\kappa}_1)\right|^2 \left|\Theta_{N_2}(\hat{z}_{2,0},\bm{\kappa}_2)\right|^2 A^{\Gamma-2}   +\cO(A^{\Gamma-3}), \quad k=0,1,2,
		\end{aligned}
	\end{equation}
	from the equation \eqref{asyexpre1}, where
	\begin{equation}\label{hus1}
	\begin{aligned}
		&\hat{x}_{1,1}^{(k,+)}(x,t)=\ii\chi_{1}^{[1]}\left( x+\chi_{1}t\right)  +h_{3,1}^{[1,k]}, \\
		&\hat{x}_{1,1}^{(k,-)}(x,t)= \ii\chi_{1}^{[1]}\left( x+\chi_{1}t\right)  -(h_{3,1}^{[1,k]})^{*}, \quad h_{3,1}^{[1,0]}=0.\\
	\end{aligned} 
	\end{equation}
    In particular, the constant term $ \Delta_{1} $ in $ (x_{0,1}, t_{0,1}) $ can be negligible as $ A\rightarrow \infty $.
   	By substituting the expression \eqref{asyexpre2} of $ \tau^{(k)} $ into the vector RW solution formula \eqref{horw}, we prove the equality of \eqref{asym_nmrp} when $\sqrt{(x-{x}_{1,0})^{2}+(t-{t}_{1,0})^{2}} = \cO(1)$ and $ A\gg1 $. It follows that, associated with every simple root $ z_{1,0} $ of the polynomial $ \Theta_{N_{1}}(z_{1},\bm{\kappa}_{1}) $, the RW solution $ q_{n}^{[N_{1},N_{2}]}(x,t) $ approaches a first-order RW $ \hat{q}_{n,l}(x-{x}_{1,0}, t-{t}_{1,0}) q_{n}^{[0]}(x,t) $. Therefore, the RW pattern of $ q_{n}^{[N_{1},N_{2}]}(x,t) $ has $ \frac{N_{1}(N_{1}-1)}{2} $ first-order RWs $ \hat{q}_{n,l}(x-{x}_{1,0}, t-{t}_{1,0}) q_{n}^{[0]}(x,t) $.
   	
   	On the other hand, when $ \sqrt{(x-{x}_{2,0})^{2}+(t-{t}_{2,0})^{2}} = \cO(1) $ and $ A\gg1 $, the asymptotic expression \eqref{asym_nmrp} with $ l=2 $ can also be established by a similar argument. In this case, the constant term $\Delta_{2}$ in $(x_{2,0},t_{2,0})$ can be computed as 
   	\begin{equation}\label{delta2}
   		\begin{aligned}
   			\Delta_{2}= \dfrac{ \left| \mathbf{P}[\mathbf{r}_{11},\mathbf{s}_{12}]\right|}{ \left| \mathbf{P}[\mathbf{r}_{11},\mathbf{s}_{11}]\right|}, 
   		\end{aligned}
   	\end{equation}
   	where the row index vector $\mathbf{r}_{11}$ is defined as in the equation \eqref{delta1}, and the column index vectors are given by 
   	$\mathbf{s}_{1m}= \mathbf{s}\big|_{\bar{\mathbf{v}}_{1}=\mathbf{v}_{1}^{(1)},\, \bar{\mathbf{v}}_{2}=\mathbf{v}_{2}^{(m)}}$ 
   	with $m=1,2$. The proof follows similarly and is omitted here. Hence, the RW pattern of $ q_{n}^{[N_{1},N_{2}]}(x,t) $ has also $ \frac{N_{2}(N_{2}-1)}{2} $ first-order RWs $ \hat{q}_{n,2}(x-{x}_{2,0}, t-{t}_{2,0}) q_{n}^{[0]}(x,t) $.  
   	
   	Therefore, we complete the proof of Theorem \ref{theo-nmrp}.

\section{Conclusions and discussions}\label{Sec-conclusion}
In this work, we constructed a class of high-order vector RW solutions for the CNLS equation \eqref{cnlse}, under the setting that the associated characteristic equation \eqref{chareq1} possesses two pairs of complex-conjugate double roots and that the solutions involve multiple large internal parameters. By taking a suitable form of these large internal parameters, we obtained a variety of previously unreported RW patterns. 

We then carried out a detailed asymptotic analysis of the resulting RW patterns. A central finding is that each RW pattern contains two regions, corresponding to the root structures of two different Adler--Moser polynomials, respectively. Each region consists of a collection of spatially separated first-order RWs, with each first-order RW corresponding to a simple root of the associated polynomial. Importantly, the first-order RWs associated with different polynomials can exhibit distinct dynamics: one group displays bright profiles, whereas the other exhibits four-petaled or dark ones. To further illustrate our findings, we presented several representative examples of RW patterns in the CNLS equation \eqref{cnlse}, including double-sector, double-heart, and mixed sector-heart structures. Among these examples, some patterns exhibit the coexistence of bright and four-petaled RWs, whereas others display the coexistence of dark and bright first-order RWs. 

In addition, this paper illustrates RW patterns only for the high-order vector RW solutions whose order index $ [N_{1},N_{2}] $ has both entries being positive integers. When $ [N_{1},N_{2}] $ contains exactly one nonzero positive integer, our approach likewise produces RW patterns corresponding to the root structure of an Adler--Moser polynomial (i.e., the other polynomial degree equals zero). In this respect, the resulting RW patterns are analogous to those in the scalar NLS equation associated with root structures of the Adler--Moser polynomial in Refs. \cite{yang2024,linh2024b,YangB2025}. A key difference, however, is that in the CNLS equation such patterns may involve three distinct first-order RW structures---namely eye-shaped, four-petaled, and anti-eye-shaped RWs---whereas the corresponding scalar NLS patterns admit only the eye-shaped RW. 

Moreover, we focus on the generic case in which two relevant Adler--Moser polynomials both possess only simple roots and share no common roots in the transformed $ (x,t) $-plane. The more intricate situation where the associated polynomials contain multiple roots will be the subject of our future work. Finally, the methodology developed here can be extended to general multi-component NLS equations and to other multi-component integrable systems, thereby enriching the landscape of RW pattern research in multi-component integrable systems and further advancing the theory of RWs.

\section*{Conflict of interests}
The authors have no conflicts to disclose.

\section*{DATA AVAILABILITY}
Data sharing is not applicable to this article as no new data were created or analyzed in this study.

\section*{Acknowledgments} 	
Liming Ling is supported by the National Natural Science Foundation of China (No. 12471236), the Guangzhou Municipal Science and Technology Project (Guangzhou Science and Technology Plan, No. 2024A04J6245) and Guangdong Natural Science Foundation
grant (No. 2025A1515011868)

%%%%%%%%%%%%
\appendix
\renewcommand\thesection{A}
\section*{Appendix A}\label{app:A}
%\refstepcounter{section}  % 手动增加 section 计数器，使其编号为 A
\addcontentsline{toc}{section}{Appendix A}
\renewcommand{\theequation}{A.\arabic{equation}}
\setcounter{equation}{0} % 重新开始编号
Here, we define two operators: the row shift operator $\mathcal{L}_1$ and the column shift operator $\mathcal{L}_2$
\begin{definition}\label{defi_1}
	Let $ N\in \mathbb{Z}_{>0} $ and $ B = \det_{1\le i,j\le N}\big(\varphi_{2j-i}\big) $.
	
	\begin{enumerate}
		\item[(1)] The operator $\mathcal{L}_1$ acts on $B$ by shifting one row at a time:
		\[
		\mathcal{L}_1 B
		:=
		\sum_{k=1}^{N}
		\det_{1\le i,j\le N}\big(\varphi_{2j-i-p_{i,k}}\big),
		\]
		where
		\[
		p_{i,k}=
		\begin{cases}
			2, & i=k,\\
			0, & i\neq k.
		\end{cases}
		\]
		Equivalently, the $k$-th term in the sum is obtained by replacing the $k$-th row
		$\big(\varphi_{2-k}, \varphi_{4-k}, \ldots, \varphi_{2N-k}\big)$ of the determinant
		$B$
		with the shifted row $\big(\varphi_{-k}, \varphi_{2-k}, \ldots, \varphi_{2N-k-2}\big)$,
		while keeping all other rows unchanged.
		
		\item[(2)] The operator $\mathcal{L}_2$ acts on $B$ by shifting one column at a time:
		\[
		\mathcal{L}_2 B
		:=
		\sum_{k=1}^{N}
		\det_{1\le i,j\le N}\big(\varphi_{2j-i-p_{j,k}}\big),
		\]
		where
		\[
		p_{j,k}=
		\begin{cases}
			2, & j=k,\\
			0, & j\neq k.
		\end{cases}
		\]
		Equivalently, the $k$-th term in the sum is obtained by replacing the $k$-th column
		$\big(\varphi_{2k-1}, \varphi_{4k-2}, \ldots, \varphi_{2k-N}\big)$ of the determinant $ B $
		with the shifted column $\big(\varphi_{2k-3}, \varphi_{4k-4}, \ldots, \varphi_{2k-N-2}\big)$,
		while keeping all other columns unchanged.
	\end{enumerate}
\end{definition}

\begin{propo}\label{prop-app1}
	Let $ N\in \mathbb{Z}_{>0} $, for the determinant $ B = \det_{1\le i,j\le N}\big(\varphi_{2j-i}\big) $, we have 
	\begin{equation}\label{app_rcop1}
		\mathcal{L}_1 B = \mathcal{L}_2 B.
	\end{equation}
	
\end{propo}

\begin{proof}
	Denote the $ k $-th term in $\mathcal{L}_1 B$ by $ B^{(k)}=\det_{1\le i,j\le N}\big(\varphi_{2j-i-p_{i,k}}\big) $ with $ 1\leq k \leq N $, and let $ C_{i,j}(B) $ denote the cofactor of the entry $ (B)_{i,j} $ in the determinant $ B $. Clearly, $ C_{k,j}(B)=C_{k,j}(B^{(k)}) $. By expanding $ B^{(k)} $ along its $ k $-th row, we obtain
	\[B^{(k)}=\sum_{j=1}^{N} \varphi_{2j-k-2} C_{k,j}(B),\]
	which implies
	\begin{equation}\label{app_rcop01}
		\mathcal{L}_1 B = \sum_{k=1}^{N}B^{(k)}= \sum_{i=1}^{N}\sum_{j=1}^{N} \varphi_{2j-i-2} C_{i,j}(B).
	\end{equation}
	
	On the other hand, we can similarly obtain the $ k $-th term in $\mathcal{L}_2 B$, that is
	\[\sum_{i=1}^{N} \varphi_{2k-i-2} C_{i,k}(B),\]
	and conclude that
	\begin{equation}\label{app_rcop02}
			\mathcal{L}_2 B =\sum_{i=1}^{N}\sum_{j=1}^{N} \varphi_{2j-i-2} C_{i,j}(B).
	\end{equation}
	
	Therefore, by combining the equations \eqref{app_rcop01} and \eqref{app_rcop02}, we can prove the equation \eqref{app_rcop1}.
\end{proof}

\bibliographystyle{elsarticle-harv}
\bibliography{RefofCNLSE}

@article{am1978,
   author = {Adler, M.  and  Moser, J.},
   title = {On a Class of Polynomials Connected with the {K}orteweg-de {V}ries Equation},
   journal = {Commun. Math. Phys.},
   volume = {61},
   pages = {1--30},
   year = {1978},
   type = {Journal Article}
}

@article{am2009,
   author = {Clarkson, P. A.},
   title = {Vortices and Polynomials},
   journal = {Stud. Appl. Math.},
   volume = {123},
   pages = {37--62},
   year = {2009},
   type = {Journal Article}
}

@article{yv2003,
   author = {Clarkson, P.A.  and  Mansfield, E.L.},
   title = {The second {P}ainlev\'e equation, its hierarchy and associated special polynomials},
   journal = {Nonlinearity},
   volume = {16},
   pages = {1--26},
   year = {2003},
   type = {Journal Article}
}

@article{linh2024a,
  title={Rogue wave pattern of multi-component derivative nonlinear {S}chr{\"o}dinger equations},
  author={Lin, Huian  and  Ling, Liming},
  journal={Chaos},
  volume={34},
  pages={043126},
  year={2024},
}

@article{lin2025vector,
  title={Vector rogue wave patterns of the multi-component nonlinear {Schr{\"o}dinger} equation and generalized mixed {Adler--Moser} polynomials},
  author={Lin, Huian and Ling, Liming},
  journal={Acta Math. Sci.},
  volume={45},
  number={6},
  pages={2478--2509},
  year={2025},
  publisher={Springer}
}

@article{linh2024b,
	author = {Lin, Huian  and  Ling, Liming},
	title = {Rogue wave patterns associated with {A}dler--{M}oser polynomials featuring multiple roots in the nonlinear {S}chr{\"o}dinger equation},
	journal = {Stud. Appl. Math.},
	volume = {154},
	pages = {e12782},
	year = {2024},
}

@article{ling2014high,
  title={High-order rogue waves in vector nonlinear {S}chr{\"o}dinger equations},
  author={Ling, Liming  and  Guo, Boling  and  Zhao, LiChen},
  journal={Phys. Rev. E},
  volume={89},
  number={4},
  pages={041201},
  year={2014},
  publisher={APS}
}

@article{ZhaoLing2016,
    author = {Zhao, Li-Chen and Guo, Boling and Ling, Liming},
    title = {High-order rogue wave solutions for the coupled nonlinear Schr{\"o}dinger equations-{II}},
    journal = {J. Math. Phys.},
    volume = {57},
    number = {4},
    pages = {043508},
    year = {2016},
    issn = {0022-2488},
}

@article{qin2023optical,
  title={Optical rogue-wave patterns in coupled defocusing systems},
  author={Qin, Yan-Hong and Ling, Liming and Zhao, Li-Chen},
  journal={Phys. Rev. A},
  volume={108},
  number={2},
  pages={023519},
  year={2023},
  publisher={APS}
}

@article{yang2024,
	author = {Yang, Bo  and  Yang, Jianke},
	title = {Rogue wave patterns associated with {A}dler--{M}oser polynomials in the nonlinear {S}chr\"odinger equation},
	journal = {Appl. Math. Lett.},
	volume = {148},
	pages = {108871},
	year = {2024},
	type = {Journal Article}
}

@article{YangB2025,
  author  = {Yang, B. and Yang, J.},
  title   = {Triangular rogue clusters associated with multiple roots of {Adler--Moser} polynomials in integrable systems},
  journal = {Physica D},
  year    = {2025},
  volume  = {483},
  pages   = {134921}
}

@article{yang2021,
	author = {Yang, Bo  and  Yang, Jianke},
	title = {Rogue wave patterns in the nonlinear {S}chr{\"o}dinger equation},
	journal = {Physica D},
	volume = {419},
	pages = {132850},
	year = {2021},
	type = {Journal Article}
}

@book{GuoB2017book,
  author    = {Guo, B. and Tian, L. and Yan, Z. and others},
  title     = {Rogue Waves: Mathematical Theory and Applications in Physics},
  year      = {2017},
  publisher = {Walter de Gruyter GmbH \& Co KG},
  address   = {Berlin},
}

@article{Akhmediev2009,
  author  = {Akhmediev, N. and Ankiewicz, A. and Taki, M.},
  title   = {Waves that appear from nowhere and disappear without a trace},
  journal = {Phys. Lett. A},
  year    = {2009},
  volume  = {373},
  number  = {6},
  pages   = {675--678}
}

@article{Kjeldsen1984,
  author  = {Kjeldsen, S. P.},
  title   = {Dangerous wave groups},
  journal = {Norwegian Maritime Research},
  year    = {1984},
  volume  = {2},
  pages   = {4--16}
}

@article{Kharif2003,
  author  = {Kharif, C. and Pelinovsky, E.},
  title   = {Physical mechanisms of the rogue wave phenomenon},
  journal = {Eur. J. Mech. B Fluids},
  year    = {2003},
  volume  = {22},
  number  = {6},
  pages   = {603--634}
}

@inproceedings{Rosenthal2005,
  author    = {Rosenthal, W. and Lehner, S.},
  title     = {Results of the {MAXWAVE} project},
  booktitle = {Proceedings of the 14th 'Aha Huliko'a Winter Workshop},
  year      = {2005},
  address   = {Honolulu, Hawaii}
}

@article{Walker2004,
  author  = {Walker, D. A. G. and Taylor, P. H. and Taylor, R. E.},
  title   = {The shape of large surface waves on the open sea and the {Draupner} {New Year} wave},
  journal = {Appl. Ocean Res.},
  year    = {2004},
  volume  = {26},
  number  = {3--4},
  pages   = {73--83}
}

@article{Draper1965,
  author  = {Draper, L.},
  title   = {{`Freak'} ocean waves},
  journal = {Marine Observer},
  year    = {1965},
  volume  = {35},
  number  = {2},
  pages   = {193--195}
}

@article{2010Efimov,
  author  = {Efimov, V. B. and Ganshin, A. N. and Kolmakov, G. V. and others},
  title   = {Rogue waves in superfluid helium},
  journal = {Eur. Phys. J. Spec. Top.},
  year    = {2010},
  volume  = {185},
  number  = {1},
  pages   = {181--193}
}

@article{Shats2010,
  author  = {Shats, M. and Punzmann, H. and Xia, H.},
  title   = {Capillary rogue waves},
  journal = {Phys. Rev. Lett.},
  year    = {2010},
  volume  = {104},
  number  = {10},
  pages   = {104503}
}

@article{Stenflo2010,
  author  = {Stenflo, L. and Marklund, M.},
  title   = {Rogue waves in the atmosphere},
  journal = {J. Plasma Phys.},
  year    = {2010},
  volume  = {76},
  number  = {3--4},
  pages   = {293--295}
}

@article{Yan2010f,
  author  = {Yan, Z.},
  title   = {Financial rogue waves},
  journal = {Commun. Theor. Phys.},
  year    = {2010},
  volume  = {54},
  number  = {5},
  pages   = {947--949}
}

@article{Chabchoub2011,
  author  = {Chabchoub, A. and Hoffmann, N. P. and Akhmediev, N.},
  title   = {Rogue wave observation in a water wave tank},
  journal = {Phys. Rev. Lett.},
  year    = {2011},
  volume  = {106},
  number  = {20},
  pages   = {204502}
}

@article{Chabchoub2012,
  author  = {Chabchoub, A. and Hoffmann, N. and Onorato, M. and others},
  title   = {Observation of hierarchy of up to fifth-order rogue waves in a water tank},
  journal = {Phys. Rev. E},
  year    = {2012},
  volume  = {86},
  number  = {5},
  pages   = {056601}
}

@article{Solli2007,
  author  = {Solli, D. R. and Ropers, C. and Koonath, P. and others},
  title   = {Optical rogue waves},
  journal = {Nature},
  year    = {2007},
  volume  = {450},
  pages   = {1054--1057}
}

@article{Kibler2010,
  author  = {Kibler, B. and Fatome, J. and Finot, C. and others},
  title   = {The {Peregrine} soliton in nonlinear fibre optics},
  journal = {Nature Physics},
  year    = {2010},
  volume  = {6},
  number  = {10},
  pages   = {790--795}
}

@BOOK{Pelinovsky2016,
  title={Extreme Ocean Waves (2nd ed.)},
  year= {2016},
  author={Pelinovsky, Efim  and  Kharif, C. },
  publisher={Springer, New York}
}

@article{fochesato2007,
  title={Numerical modeling of extreme rogue waves generated by directional energy focusing},
  author={Fochesato, Christophe  and  Grilli, S.  and  Dias, F.},
  journal={Wave motion},
  volume={44},
  number={5},
  pages={395--416},
  year={2007},
  publisher={Elsevier}
}

@book{YangB2024book,
  author    = {Yang, B. and Yang, J.},
  title     = {Rogue Waves in Integrable Systems},
  year      = {2024},
  publisher = {Springer Nature Switzerland},
}

@article{ankiewicz2010d,
  title={Discrete rogue waves of the {A}blowitz--{L}adik and {H}irota equations},
  author={Ankiewicz, Adrian  and  Akhmediev, Nail  and  Soto-Crespo, JM},
  journal={Phys. Rev. E},
  volume={82},
  number={2},
  pages={026602},
  year={2010},
  publisher={APS}
}

@article{suh2023,
title = {Rogue waves and their patterns for the coupled {F}okas--{L}enells equations},
author={Ling, Liming  and  Su, Huajie},
journal = {Physica D},
volume = {461},
pages = {134111},
year = {2024},
ISSN = {0167-2789}
}

@article{Kedziora2013,
  author  = {Kedziora, D. J. and Ankiewicz, A. and Akhmediev, N.},
  title   = {Classifying the hierarchy of nonlinear-{S}chr{\"o}dinger-equation rogue-wave solutions},
  journal = {Phys. Rev. E},
  year    = {2013},
  volume  = {88},
  number  = {1},
  pages   = {013207}
}

@ARTICLE{zhangg2022,
  title={Rogue waves and their patterns in the vector nonlinear {S}chr{\"o}dinger equation},
  author={Zhang, G. X. and Huang, P. and Feng, B. F. and Wu, C. F.},
  journal={J. Nonlinear Sci.},
  volume={33},
  pages={116},
  year={2023},
  publisher={Springer}
}

@ARTICLE{ZhangGQ2021,
  title={Multi-component nonlinear Schr{\"o}dinger equations with nonzero boundary conditions: higher-order vector {P}eregrine solitons and asymptotic estimates},
  author={Zhang, G. Q.  and  Ling, L. M.  and  Yan, Z. Y.},
  journal={J. Nonlinear Sci.},
  volume={31},
  number={5},
  pages={81},
  year={2021},
  publisher={Springer}
}

@article{DongJ2022,
   author = {Dong, J. and Ling, L.  and  Zhang, X.},
   title = {{K}adomtsev-{P}etviashvili equation: {O}ne-constraint method and lump pattern},
   journal = {Physica D},
   volume = {432},
   year = {2022},
   type = {Journal Article}
}

@article{behrend2013,
  title={Multiply-refined enumeration of alternating sign matrices},
  author={Behrend, Roger E},
  journal={Adv. Math.},
  volume={245},
  pages={439--499},
  year={2013},
  publisher={Elsevier}
}


%\begin{thebibliography}{00}


%\bibitem{}

%\end{thebibliography}
\end{document}